\documentclass[11pt]{article}

\usepackage{amsmath}
\usepackage{amssymb}
\usepackage{fullpage}
\usepackage{color}
\usepackage{graphicx}
\usepackage{epsfig}
\usepackage{amsthm}
\usepackage{latexsym}
\usepackage{amssymb}
\usepackage{amsmath}
\usepackage{verbatim}
\usepackage{graphicx}
\usepackage{epsfig}
\usepackage{subfigure}
\usepackage{cite}
\usepackage{url}
\usepackage[table]{xcolor}
\usepackage{tabularx}
\usepackage{bbm}
\usepackage{algorithmic}
\usepackage{algorithm}
\usepackage{mathtools}

\usepackage{textcomp,setspace}

\usepackage{hyperref}
\usepackage[all]{hypcap}

\newtheorem{theorem}{Theorem}
\newtheorem{proposition}{Proposition}
\newtheorem{claim}{Claim}
\newtheorem{lemma}{Lemma}
\newtheorem{corollary}{Corollary}
\newtheorem{definition}{Definition} 

\newtheorem{remark}{Remark}

\DeclareMathOperator*{\argmax}{arg\,max}
\DeclareMathOperator*{\argmin}{arg\,min}

\newcommand{\E}{\mathbb{E}}
\newcommand{\A}{\mathcal{A}}
\newcommand{\Pa}{\mathcal{P}}
\newcommand{\Q}{\mathcal{Q}}
\newcommand{\R}{\mathcal{R}}

\newcommand{\ALG}{\textsc{AlgEG}}
\newcommand{\ALGS}{\textsc{RoundRobin}}
\newcommand{\ALGBL}{\textsc{AlgBL}}
\newcommand{\ALGC}{\textsc{AlgChores}}
\newcommand{\B}{\mathcal{B}}

\newcommand{\ALGF}{\textsc{AlgSub}}

\providecommand{\myfloor}[1]{\left \lfloor #1 \right \rfloor }

\newtheorem{innercustomgeneric}{\customgenericname}
\providecommand{\customgenericname}{}

\begin{document}

\title{{\bfseries Approximation Algorithms for Maximin Fair Division}}
	\author{Siddharth Barman\thanks{Indian Institute of Science. {\tt barman@iisc.ac.in}} \and Sanath Kumar Krishnamurthy\thanks{Stanford University. {\tt sanathsk@stanford.edu}}} 
	\date{}
	\maketitle

\begin{abstract}
	We consider the problem of allocating indivisible goods \emph{fairly} among $n$ agents who have additive and submodular valuations for the goods. Our fairness guarantees are in terms of the \emph{maximin share}, that is defined to be the maximum value that an agent can ensure for herself, if she were to partition the goods into $n$ bundles, and then receive a minimum valued bundle. Since maximin fair allocations (i.e., allocations in which each agent gets at least her maximin share) do not always exist, prior work has focused on approximation results that aim to find allocations in which the value of the bundle allocated to each agent is (multiplicatively) as close to her maximin share as possible. In particular, Procaccia and Wang (2014) along with Amanatidis et al. (2015) have shown that under additive valuations a $2/3$-approximate maximin fair allocation always exists and can be found in polynomial time. We complement these results by developing a simple and efficient algorithm that achieves the same approximation guarantee. 

Furthermore, we initiate the study of approximate maximin fair division under \emph{submodular valuations}. Specifically, we show that when the valuations of the agents are \emph{nonnegative}, \emph{monotone}, and submodular, then a $0.21$-approximate maximin fair allocation is guaranteed to exist. In fact, we show that such an allocation can be efficiently found by using a simple round-robin algorithm. A technical contribution of the paper is to analyze the performance of this combinatorial algorithm by employing the concept of \emph{multilinear extensions}. 
\end{abstract}

\section{Introduction}

This problem of allocating goods among agents in a ``fair'' manner has been extensively studied in the economics and mathematics literature for over fifty years; the first mathematical treatment of fair division appears in the work of Steinhaus, Banach, and Knaster~\cite{first-attempt}. Such problems naturally arise in many real-world settings---e.g., in government auctions, border disputes, and divorce settlements---and a vast literature has been developed to address fair division under various modeling assumptions. Multiple solution concepts, such as \emph{envy-freeness}, \emph{proportionality}, and \emph{equitability}, have been proposed to formally capture fairness; see~\cite{brams1996fair} and~\cite{moulin2004fair} for excellent expositions. 

Even though there is a significant body of work aimed at understanding the division of \emph{divisible} goods,\footnote{Such goods model resources that can be fractionally assigned, e.g., land.} questions related to the fair division of \emph{indivisible} goods\footnote{An indivisible good represents a discrete resource which cannot be fractionally assigned among different agents, i.e., each such good must be assigned to a single agent.} are relatively underexplored. In particular, it is known that fair allocations of divisible goods necessarily exist (with respect to all the previously mentioned solution concepts), however, such guarantees do not hold for indivisible goods; for example, if we have two agents and a single indivisible good, then, under any allocation, the agent that does not get the good will \emph{envy} the other. Furthermore, in the case of indivisible goods, even multiplicative approximations of classical fairness notions (like envy-freeness and proportionality) cannot be achieved.

Motivated by these observations and the fact that many practical scenarios (such as course allocation at educational institutions~\cite{budish-approx-ceei}) inherently entail allocation of discrete resources, recent results have focussed on the problem of fair division of indivisible goods; see, e.g., \cite{budish-approx-ceei, bouveret-basic, EC-2/3, ICALP-2/3} and references therein. Specifically, this thread of research has considered notions of fairness\footnote{For example, \emph{EF$1$} \cite{budish-approx-ceei}, \emph{EFX} \cite{unreasonalble-nsw}, and \emph{maximin shares}~\cite{budish-approx-ceei}.} which are applicable to indivisible goods, and established existence and computational results for these solution concepts. 

One such fairness notion is the \emph{maximin share}, which is defined in the notable work of Budish \cite{budish-approx-ceei} (see also~\cite{moulin1990uniform}). Conceptually, maximin share can be interpreted through a discrete (indivisible) generalization of the standard \emph{cut-and-choose} protocol. In particular, this protocol is used to fairly divide a cake (i.e., a divisible good) between two agents: the first agent cuts the cake and then the second agent gets to select her favorite piece. Note that a risk-averse agent would cut the cake into two pieces that have equal value for her. In other words, the first agent will cut (partition) the cake to maximize the minimum value over the pieces in the cut (partition). The protocol leads to a fair allocation in terms of {envy-freeness}, i.e., it results in an allocation wherein each agent prefers her piece of the cake over the other agent's piece. 

Analogously, while allocating indivisible goods among $n$ agents, we can (hypothetically) ask an agent $i$ to partition the goods into $n$ bundles, and then the other agents get to pick a bundle before $i$. This could possibly lead to agent $i$ getting her least desired bundle in the partition. Hence, as in the cut-and-choose protocol, a risk-averse agent  would partition the goods so as to maximize the value of the least desirable bundle (according to her) in the partition. The value that agent $i$ can guarantee for herself by partitioning in this manner is defined to be agent $i$'s \emph{maximin share}. An allocation is said to be \emph{maximin fair} if, for every agent, the value of her bundle is at least as high as her maximin share.

Overall, the maximin share provides an intuitive threshold, and using it, we can deem an allocation to be fair if this threshold is met for every agent. Recent results in the computer science literature have considered whether maximin fair allocations exist and whether they can be efficiently computed. It turns out that the existence of a maximin fair allocation is not guaranteed~\cite{EC-2/3, kurokawa2016can}. However, this concept lends well to approximation guarantees: the work of Procaccia and Wang \cite{EC-2/3} shows that when the agents have \emph{additive valuations}, one can find an allocation wherein each agent gets a bundle of value greater than $2/3$ times her maximin share, i.e., a $2/3$-approximate maximin fair allocation always exists. The proof of existence by Procaccia and Wang is constructive, but it leads to a polynomial-time algorithm only when the number of agents is constant. Extending this result, Amanatidis et al. \cite{ICALP-2/3} have developed an algorithm that finds a $2/3$-approximate maximin fair allocation and runs in time polynomial in the number of agents. It is relevant to note that the approximation guarantees obtained in these results and the current paper are absolute, i.e., they guarantee that there \emph{always} exists an allocation wherein each agent receives a value at least a constant times her maximin share. The guarantees are not relative to the ``best-possible'' allocation.  

In the context of indivisible goods, the above-mentioned results affirm the relevance of maximin shares as a measure of fairness, but they are confined to additive valuations. In this paper, we extend this line of work and establish maximin fairness guarantees for both \emph{additive} and \emph{submodular} valuations.\footnote{These valuation classes are formally defined in Section~\ref{Notations}.} Given that submodular functions are extensively used in economics and computer science to model preferences, our fairness guarantee for submodular valuations is a compelling generalization of prior work. 

\subsection{Our Results and Techniques}
\noindent
{\bf Additive Valuations:} Our first result (Theorem~\ref{thm:main}) considers fair division of indivisible goods between agents with additive valuations for the goods. In particular, we show that the approximation guarantee obtained in~\cite{EC-2/3} and~\cite{ICALP-2/3} can be achieved by a simple,\footnote{Simple in the sense that our algorithm just entails sorting values and finding cycles in directed graphs. By contrast, the algorithm in~\cite{ICALP-2/3} uses as a subroutine a PTAS (of Woeginger \cite{woeginger1997polynomial}) which is based on solving integer linear programs of constant dimension, using Lenstra's algorithm~\cite{lenstra1983integer}.} combinatorial algorithm.\footnote{In fact, we are able to compute a $\frac{2n}{3n-1}$ approximate maximin fair allocation in polynomial time. This is slightly better than the $(\frac{2}{3}-\epsilon)$ approximation guarantee of Amanatidis et al. \cite{ICALP-2/3}.} That is, we present a polynomial-time algorithm that finds a $2/3$-approximate maximin fair allocation under additive valuations. 

Our algorithm for additive valuations consists of two parts. First (in Algorithm~\ref{alg:bouveret}: $\ALGBL$), we reduce the problem of finding a maximin fair allocation to a restricted setting where the agents value goods in the same order, i.e., in this constructed instance all the goods can be ordered (indexed) such that the valuation of every agent respects this ordering. This reduction was proposed by Bouveret and Lema\^{i}tre~\cite{bouveret-basic} and it preserves (approximate) maximin fairness in the sense that the $g$th good in the reduced instance corresponds to the right to select a good in round $g$. In particular, given an (approximately) maximin fair allocation, say $\A$, of the reduced instance one can construct a sequence $\mathcal{S}$ and use $\mathcal{S}$ to assign goods in the original instance: the agent with the $g$th good in $\A$ is at the $g$th position in $\mathcal{S}$ and, hence, gets to greedily select a good in the $g$th round. For additive valuations, this reduction ensures that if $\A$ is an $\alpha$-approximate maximin fair allocation in the reduced instance, then the allocation resulting from $\mathcal{S}$ is an $\alpha$-approximate maximin fair allocation
for the original instance. 

With this reduction in hand, we focus on this restricted (ordered) setting and develop an efficient algorithm (Algorithm~\ref{alg:envygraph}: $\ALG$) to find $2/3$-approximate maximin fair allocations for ordered instances. To obtain this result we rely on a careful modification of an algorithm by Lipton et al. \cite{lipton-envy-graph}, which, in particular, finds an allocation which is \emph{envy-free up to one good} (EF1); formal definitions appear in Section~\ref{Notations}. An arbitrary EF1 allocation might not provide any non-trivial approximation guarantees for general instances; see Appendix~\ref{example:ef1:mms} for an example. However, we show that for ordered instances an instantiation of the algorithm by Lipton et al.~\cite{lipton-envy-graph} (in a particular manner) finds an allocation with the desired maximin fairness guarantee.

The algorithm of Lipton et al.~\cite{lipton-envy-graph} employs a construct called the \emph{envy graph}. This directed graph captures envy between agents and using it goods are iteratively assigned to the agents. Throughout the assignment process the invariant that the envy graph is acyclic is explicitly maintained--if at any point of time an envy cycle is created, then the agents can swap their bundles around the cycle to monotonically increase their values and resolve the envy. The goods are assigned in order and each good is allocated to a currently unenvied agent; since the envy graph is acyclic such an agent (i.e., a source node) is guaranteed to exist. We show that for ordered instances this iterative process leads to an allocation that is, in fact, \emph{envy-free up to the least valued good} (EFX); see Definition~\ref{defn:efx}. EFX is a stronger notion than EF1 and it ensures that any agent who owns three or more goods is not envied beyond a multiplicative factor of $3/2$.\footnote{Note that, in the additive context, a multiplicative bound in terms of envy implies a matching bound with respect to the maximin shares as well.}  The rest of the proof entails a careful analysis of allocations in which each bundle has at most two goods. We perform this analysis using the concept of \emph{majorization} (see Section~\ref{section:majorization} for a definition) and, overall, show that a $2/3$-approximate maximin fair allocation can be computed efficiently when the agents' valuations are additive. \\

\noindent
{\bf Submodular Valuations:} We initiate the study of approximate maximin fair division under submodular valuations. Specifically, we show that when the valuations of the agents are \emph{nonnegative}, \emph{monotone}, and submodular, then a $0.21$-approximate maximin fair allocation is guaranteed to exist and, in fact, such an allocation can be efficiently computed via a simple round-robin algorithm (Theorem~\ref{thm:mms-const}). 

A technical contribution of the paper is to analyze the performance of this simple algorithm by employing the concept of \emph{multilinear extensions}. This concept has been used in recent results for constrained submodular maximization and, at a high level, it can be thought of as a continuous surrogate (a continuous extension) for a submodular function; see Section~\ref{sec:Multilinear} for a definition. It is observed in~\cite{vondrak-matroid} that the multilinear extension---of an agent's valuation function---achieves a high value if the agent is allocated a bundle at random. We use this observation to show that, for each of the $n$ agents, a bundle picked at random approximately satisfies the maximin requirement in expectation; specifically, for each agent $i$, the expected value of a random subset (wherein every good is included independently with probability $1/n$) is at least $(1 - \frac{1}{e})$ times the maximin share of $i$. However, to obtain the desired approximation guarantee from this fact we cannot directly apply standard techniques (such as {pipeage rounding or randomized  rounding}), since we need to satisfy the maximin requirement for multiple submodular functions simultaneously.\footnote{Note that, a direct application of randomized rounding leads to a logarithmic approximation guarantee. In particular, the concentration bounds for submodular functions~\cite{conc-submod} only establish the existence of a $\frac{1}{O(\log n)}$-approximate maximin fair allocation.} Instead, we employ a round-robin method with a preprocessing step. 

Since the problem of (exactly) computing the maximin share is {\rm NP}-hard,\footnote{The problem is computationally hard even for additive valuations and two agents, since the Partition problem reduces to it.} our algorithm works with estimates---$\tau_i$ for agent $i$---of the maximin shares. Each estimate $\tau_i$ is initialized to be greater than the corresponding maximin share and we reduce these estimates iteratively--the properties of our method ensure that we never make them significantly smaller than the maximin shares.  

In our algorithm we first perform a preprocessing step wherein high-valued goods are assigned as singleton bundles. In particular, if for an agent $i$ there exists an unallocated good of value at least $\frac{1}{3} \left( 1 - \frac{1}{e} \right) \tau_i \approx 0.21 \tau_i$, then we allocate the good to the agent and remove $i$ (and the now assigned good) from consideration. The leftover (low-valued) goods are partitioned among the remaining agents (who have not received a good yet) in a round-robin manner. Here, in each round, the participating agents greedily select an unallocated good one after the other, i.e., in each round, every agent $i$ is assigned a good (from the set of unallocated goods) that leads to the maximum possible increase in $i$'s valuation. 

Note that the stated fairness guarantee holds directly for agents that receive a high-valued good.  In addition, recall that, under submodular valuations, a uniformly random allocation is $(1 - 1/e)$-maximin fair, in expectation; specifically, if $G^0$ denotes the set of goods that remain unallocated after the preprocessing step and, for agent $i$, $E^0$ is the expected value of a random subset (sampled as mentioned above) of $G^0$, then $E^0$ is at least $(1-1/e)$ times $i$'s maximin share. Given this bound, we analyze the valuation of the bundles assigned to the agents (by the round-robin method) in terms of this expected value. 

Write $G^r$ to denote the set of goods that remain unallocated after the $r$th round of the round-robin phase and $E^r$ be the expected marginal value of a random subset of $G^r$. At a high level, we show that, in each round $r$ and for every agent $i$, the marginal gain experienced by $i$ (as a result of greedily selecting a good) is somewhat comparable to $E^r - E^{r+1}$. Since, for each agent $i$, the sum of the marginal gains (across all the rounds) is equal to the value of the bundle finally assigned to $i$, employing a telescoping sum over $E^r - E^{r+1}$ (and balancing out the parameters) we get that the round-robin algorithm allocates to each agent a bundle of value at least $\frac{1}{3} E^0$, i.e., a bundle of value at least $\frac{1}{3} \left(1 - \frac{1}{e} \right) \approx 0.21$ times her maximin share. 

\begin{remark}
Since additive functions are also submodular, the result developed for submodular valuations also applies to the additive case. However, for the additive setting we get a stronger approximation guarantee. It is also relevant to note that our algorithms---and their analysis---for these two valuations classes are quite distinct. 
\end{remark}
 
\noindent
{\bf Chores (negatively valued goods):}
We also show that a $4/3$-approximate maximin fair allocation can be efficiently computed when agents have additive valuations for \emph{chores} (i.e., negatively valued goods). This improves upon the previously best-known $2$-approximation of Aziz et al.~\cite{aziz-chores}; see Appendix~\ref{proof:chores}.\footnote{In the case of chores the approximation factor is greater than or equal to one.}

Our algorithm for fair division of chores is very similar to the above-mentioned algorithm for (positively valued) goods. As in the goods case, the algorithm for chores has two parts. The first part uses the reduction of Bouveret and Lema\^{i}tre~\cite{bouveret-basic} to obtain ordered instances wherein agents value chores in the same order.  The second part is analogous to the goods setting as well. In particular, it employs an envy-graph algorithm with the only difference that chores are allocated to sink nodes in the envy graph. 

Since the valuations are additive in the chores context as well, the analysis follows suit. Specifically, with respect to bundles with three or more chores, we show that envy is bounded up to a relatively small (in terms of absolute value) chore. In addition, we have to carefully analyze allocations in which bundles have at most two chores. In comparison to the goods case, this  analysis leads to a better bound for chores. Hence, we obtain a stronger approximation guarantee for negatively valued goods.

\subsection{Related Work} 
As mentioned previously, fair division is a well-studied problem in multiple disciplines~\cite{brams1996fair, moulin2004fair, moulin2016handbook}. In this paper we focus on allocating indivisible goods and, in particular, study the concept of maximin shares. This concept was defined in~\cite{budish-approx-ceei} as a natural relaxation of \emph{fair-share guarantees}~\cite{moulin2014cooperative}. In particular, Budish \cite{budish-approx-ceei} applied this solution concept to the problem of allocating courses to students, and showed that maximin shares for $n$ agents can be guaranteed if some goods are over allocated  and the share of each agent is computed with respect to $n+1$ agents, instead of $n$. Note that these allocation errors cannot be ignored when, say, there is only one copy of each good and, hence, such a bi-criteria approximation guarantee does not provide any nontrivial multiplicative approximation bounds.

The work of  Bouveret and Lema\^{i}tre~\cite{bouveret-basic} focuses on fair division of indivisible goods with additive valuations. Along with other results, they show that if an allocation is fair with respect to other solution concepts (such as envy-freeness and proportionality), then it is maximin fair as well. They also show that maximin fair allocations are guaranteed to exist for binary, additive valuations. The empirical results presented in~\cite{bouveret-basic} suggest that maximin shares invariably exist when the additive valuations are drawn from particular distributions. 

Procaccia and Wang \cite{EC-2/3} also study maximin shares under additive valuations. They prove that maximin fair allocations do not always exist, but allocations wherein each agent gets a bundle of value greater than $2/3$ times her maximin share is guaranteed to exist. The proof of existence in~\cite{EC-2/3} is constructive, but it leads to a polynomial time algorithm only when the number of agents is constant. Amanatidis et al.~\cite{ICALP-2/3} extend this result by showing that a $2/3$-approximate maximin fair allocation can be computed in polynomial time. Kurokawa et al.~ \cite{kurokawa2016can} strengthen the negative result of Procaccia and Wang~\cite{EC-2/3}: they present smaller fair division instances which do not admit maximin fair allocations. In addition, they also show that when the valuations are randomly drawn then maximin fair allocations exist with high probability. 

Aziz et al. \cite{aziz-chores} consider maximin fair allocations for chores (i.e., negatively valued goods). They show that, even in the case of chores, exact maximin fair allocations do not always exist and they complement this result by developing an efficient algorithm which finds a $2$-approximate maximin fair allocation.

Maximin fair division is seemingly related to the well-studied Santa Claus problem~\cite{bezakova2005allocating, bansal2006santa, feige2008allocations, annamalai2014combinatorial, asadpour2010approximation, chakrabarty2009allocating}, where the goal is to find an allocation that maximizes the minimum value over all agents.\footnote{The Santa Claus problem is also called the max-min fair allocation problem.} However, note that the Santa Claus problem is an {optimization} problem---the objective function is the egalitarian social welfare---whereas, in the maximin fair division problem the goal is to find an allocation that satisfies a {property} (that every agent gets her maximin share). The two problems behave differently with respect to scaling. In particular, a maximin fair allocation continues to be fair if a single agent scales her valuations, this is not the case with the Santa Claus problem. These problems differ significantly in terms of approximability as well: In the additive valuation case, the best known algorithm for the Santa Claus problem achieves a $\widetilde{O}(n^{\varepsilon})$-approximation and runs in time $O(n^\frac{1}{\varepsilon})$~\cite{chakrabarty2009allocating}; here $n$ is the number of agents and $\varepsilon >0$.  For the submodular valuation case, a $O(n^{1/4} m^{1/2})$-approximation algorithm was developed for the Santa Claus problem by Goemans et al.~\cite{goemans2009approximating}; here $m$ is the number of goods. On the other hand, as we show in this paper, constant-factor approximation guarantees can be achieved for the maximin fair division problem, even when the valuations are submodular.  

In notable cases, fair division algorithms have been implemented: (i) Course Match~\cite{budish2016course} is used for course allocation at Wharton School at the University of Pennsylvania, (ii) the website Spliddit~\cite{spliddit} provides free access to fair division methods, and (iii) the Adjusted Winner Website\footnote{\url{http://www.nyu.edu/projects/adjustedwinner/}} implements an algorithm of Brams and Taylor \cite{brams1996fair} to fairly divide goods between two players. These practical applications enforce the idea that simple/easily implementable fair division algorithms---like the ones developed in this paper---have potential for impact.

\subsection{Subsequent Work}
In a recent result Ghodsi et al.~\cite{ghodsi2017fair} have obtained improved approximation guarantees for maximin fair division under additive and submodular valuations. They also consider more general valuation classes. Specifically, Ghodsi et al.~\cite{ghodsi2017fair} obtain a $3/4$-approximation for additive valuations, a $1/3$-approximation for submodular valuations, a constant-factor approximation for \emph{XOS valuations}, and a logarithmic approximation for \emph{subadditive valuations}. For the additive case, the algorithm of Ghodsi et al.~\cite{ghodsi2017fair} requires intricate preprocessing steps: their algorithm first clusters the agents (based on their valuation functions) and then applies ``bag-filling'' procedures to each cluster. Furthermore, variants of envy-freeness are devised in~\cite{ghodsi2017fair} to obtain the stated approximation guarantee for additive valuations. In~\cite{ghodsi2017fair}, the results for more general valuation classes are obtained by combining the valuations of the agents into a single function which achieves a high value for an allocation iff the allocation is approximately maximin fair. Approximation algorithms (based on local-search methods) for maximizing this unified function essentially lead to the stated results for submodular, XOS, and subadditive valuations. 

An extended version of~\cite{EC-2/3}---specifically,~\cite{kurokawa2018fair}---also provides a $2/3$ approximation guarantee for additive valuations; the bound is slightly better when the total number of agents is even. The approach in~\cite{kurokawa2018fair} is similar to the one developed in the present paper, in the sense that both the results employ the envy-graph algorithm of~\cite{lipton-envy-graph}. However, Kurokawa et al.~\cite{kurokawa2018fair} do not use the reduction of~\cite{bouveret-basic} and, instead, rely on a series of preprocessing steps which, in particular, bound the two highest valued goods in each bundle. Note that a similar structure appears in this work as well.

\section{Notation and Preliminaries} \label{Notations}
We let $[n]:=\{1,2,\ldots,n\}$ denote the set of agents and $[m]:=\{1,2,\ldots, m\}$ denote the set of indivisible goods. The valuation of an agent $i$ for a subset of goods  $S \subseteq [m]$ is denote by $v_{i}(S)$. For ease of presentation, we will use $v_i(g)$---instead of $v_i(\{ g \})$---to denote the valuation of agent $i$ for good $g \in [m]$.

Agents are said to have \emph{additive} valuations when, for any subset of goods $S \subseteq [m]$, the valuation of agent $i$ satisfies $v_i(S) := \sum_{g \in S} v_{i}(g) $. For ease of presentation, in the case of additive valuations, we will also use $v_{i,g}$ to denote  $v_{i}(g)$. The current work primarily addresses \emph{nonnegative} valuations, $v_{i,g}\geq 0$ for all  $i\in[n]$ and $ g\in[m]$. Results for negatively valued goods (i.e., chores) are confined to Appendix~\ref{proof:chores}. 

Agents are said to have nonnegative, \emph{monotone}, \emph{submodular} valuations (with $v_i(\emptyset)=0$) when for any $A \subseteq B \subset [m]$ and $g \in [m]\setminus B$, the valuation of agent $i$ satisfies:
\begin{align*}
v_i(\emptyset) & = 0   \\
v_i(A)  & \leq v_i(B) \qquad \text{(Monotonicity)} \\ 
v_i(B\cup\{g\}) - v_i(B)  & \leq v_i(A\cup\{g\}) - v_i(A)  \qquad \text{(Submodularity)}
\end{align*} 

Given a set function $f: 2^{[m]} \mapsto \mathbb{R}$ and a subset $H \subseteq [m]$, we will consider \emph{marginal} values with respect to $H$ and, in particular, define the marginal function $f_H$ as follows: $f_H(S):= f(H\cup S)- f(H)$ for all $ S\subseteq [m]$. Note that, for any $H \subseteq [m]$, if $f$ is submodular then so is $f_H$.

Write $\Pi_n(S)$ to denote the set of all \emph{$n$-partitions} of set $S \subseteq [m]$. With a slight abuse of notation, we will write $\Pi_n(m)$ to denote $\Pi_n([m])$. Throughout, we will use the word \emph{allocation} to denote an $n$-partition $(A_1, \ldots, A_n)$ where set $A_i$ is assigned to agent $i$. In addition, an allocation of a strict subset of the goods will be referred to as a \emph{partial allocation}. We will use the term \emph{bundles} to denote subsets, $A_i$s, in an (partial) allocation.

Our fairness guarantee is in terms of \emph{maximin shares}. Formally, 

\begin{definition}[Maximin Share]
	\label{defn:maxmin}
	For an agent $i \in [n]$ and a subset of goods $S \subseteq [m]$, the $n$-maximin share is defined to be 
	\begin{align*}
	\mu_i^n(S) & := \max_{ (M_1, M_2, \ldots, M_n) \in \Pi_n(S)} \  \min_{k \in [n]} v_i(M_k).
	\end{align*}
\end{definition}

For ease of presentation, we will use $\mu_i$ to denote $\mu_i^n([m])$ and use the term maximin share to refer to the $n$-maximin share, whenever $n$ is clear from the context. Ideally, we would like to ensure fairness by partitioning the goods such that each agent gets her maximin share, i.e., partition the goods into subsets  $(A_1, A_2, \ldots, A_n) \in \Pi_n(m)$ such that  $v_i(A_i) \geq \mu_i$ for all $ i \in [n]$.

Since such partitions do not always exist (see Appendix~\ref{example:mms:existance} for an example), a natural goal is to study approximation guarantees. In particular, our objective is to develop efficient algorithms that determine a partition $(A_1, \ldots, A_n) \in \Pi_n(m)$ wherein each agent, $i$, gets a bundle, $A_i$, of value (under $v_i$) at least $\alpha$ times her maximin share, with $\alpha \in (0,1]$ being as large as possible. We call such partitions as \emph{$\alpha$-approximate maximin fair} allocations. When $\alpha=1$, we say that the allocation is maximin fair.

In this work we develop simple and efficient algorithms that achieve this objective with $\alpha = 2/3$ for additive valuations and $\alpha = 0.21$ for submodular valuations.

Our analysis of the algorithm (for additive valuations) that achieves this result relies on understanding the fairness of an allocation in terms of \emph{envy}. Formally, for an allocation $\B= (B_1, \ldots, B_n)$, we say that agent $i$ \emph{envies} agent $j$ iff $i$ prefers $j$'s bundle to her own, i.e., $v_i(B_i) < v_i (B_j)$. An allocation is defined to be \emph{envy-free} iff no agent envies any other. Since envy-free allocations are not guaranteed to exist with indivisible goods,\footnote{Consider a setting with two players and one indivisible good that has a nonzero value for both the players.} a natural relaxation (which has been studied in literature; see, e.g.,~\cite{budish-approx-ceei}) is \emph{envy-free up to one good} (EF1): an allocation $\A$ is said to be EF1 iff for each $i, j \in [n]$ there exists a good $g \in A_j$ such that $ v_i(A_i) \geq v_i(A_j \setminus \{ g \})$. In fact, we will show that---for a relevant class of problem instances---our algorithm finds an allocation which is \emph{envy-free up to the least valued good} (EFX) (defined in \cite{unreasonalble-nsw}).
\begin{definition}[Envy-free up to the least valued good (EFX)]
\label{defn:efx}
	An allocation $\A$ is said to be envy-free up to the least valued good (EFX)  iff for every $i, j \in [n]$ and each good $g \in A_j$ we have $ v_i(A_i) \geq v_i(A_j \setminus \{g\})$.
\end{definition}

\section{Additive Valuations}
\label{sect:additive}

In this section we present an algorithm that efficiently finds a $\frac{2}{3}$-approximate maximin fair allocation under additive valuations. As mentioned previously, the existence and efficient computation of $\frac{2}{3}$-approximate maximin fair allocations (under additive valuations) has been addressed in prior work~\cite{EC-2/3, ICALP-2/3}. The key contribution of this section is to show that such an approximation guarantee can, in fact, be achieved via a simple, combinatorial algorithm. 

\begin{theorem}[Main Result for Additive Valuations] \label{thm:main}
Given a set of $n$ agents with additive valuations, $\{v_i\}_{i \in [n]}$, for a set of $m$ indivisible goods, we can find a partition $(A_1, \ldots, A_n) \in \Pi_n(m)$ in polynomial time that satisfies 
\begin{align*}
v_i(A_i) \geq \frac{2n}{3n-1} \ \mu_i \qquad \text{ for all } i \in [n].
\end{align*}
Here $\mu_i$ is the maximin share of agent $i$.
\end{theorem}

The proof of the theorem proceeds in two parts: (i) First, we show that the problem of finding an approximate maximin fair allocation can be reduced to a restricted setting where the agents value the goods in the same order. That is, the $m$ goods can be ordered (indexed), say $g_1, g_2, \ldots, g_m$, such that the valuation of every agent $i$ respects this ordering: for each  $a < b$ we have $ v_{i} (g_a) \geq v_{i} (g_b)$. (ii) Then, we develop a $2/3$-approximation algorithm for this restricted setting. 


\subsection{The Reduction of Bouveret and Lema\^{i}tre} \label{section:bouveret}
Intuitively, the greater the conflict of interest among agents, the harder it is to obtain a fair allocation. For additive valuations, Bouveret and Lema\^{i}tre \cite{bouveret-basic} show that it is hardest to guarantee the agents their maximin share when all the agents have the same order of preference over the goods. Formally, we say an instance is an \emph{ordered instance} iff there exists a total ordering ($\prec$) over the set of goods $[m]$ such that for all  $i\in[n]$ and $ g, g' \in [m]$, such that  $g\prec g'$, we have  $v_{i,g}\geq v_{i,g'}$. Therefore, without loss of generality, we can say that an instance is ordered iff for all agents $i\in[n]$ and goods $ g, g'\in[m]$, such that $g< g'$, we have $v_{i,g} \geq v_{i,g'}$, i.e., $v_{i,1}\geq v_{i,2}\geq \dots \geq v_{i,m}$.

Given a fair division instance $I=\left([n],[m],\{v_i\}_{i \in [n]} \right)$ we will construct an ordered instance as follows.  First, note that for every agent $i \in [n]$, there exists a permutation $\sigma_i: [m] \rightarrow [m]$ such that for all  $g, g' \in [m]$ with $g < g'$ we have  $v_{i,\sigma_i(g)} \geq  v_{i,\sigma_i(g')}$. Using these permutations, we define a new valuation function, $v'_i$, for every agent $i$ by setting $v'_{i,g} = v_{i,\sigma_i(g)}$ for all $ g \in [m]$. That is, for agent $i$, the value of the $g$th good in the new instance is equal to the $g$th largest valuation of $i$ in the original instance. We refer to $I'=\left([n],[m],\{v'_i\}_{i \in [n]} \right)$ as the ordered instance of $I$. Note that we can find the ordered instance of $I$ in $O(n m \log m)$ time.

Bouveret and Lema\^{i}tre~\cite{bouveret-basic} show that if there is an allocation $\mathcal{A}'$ which guarantees every agent her maximin share in $I'$, then there is an allocation $\mathcal{A}$ that guarantees every agent her maximin share in $I$. Moreover, given a maximin fair allocation $\mathcal{A}'$ (for $I'$), a maximin fair allocation $\mathcal{A}$ for the original instance $I$ can be found in polynomial time. In fact, we directly use their proof to establish the following slightly stronger statement that such a reduction is possible with respect to any set of scalars, $\alpha_i$s, and not just with respect to the maximin shares.  

{
	\begin{algorithm}
		{{\bf Input :} Instance $I= \left( [n],[m], \{v_i \}_{i \in [n]} \right)$ with an allocation $\mathcal{A}'=\left( A'_1, \ldots, A'_n\right)$ for the ordered instance $I'=\left( [n],[m], \{v'_i\}_{i \in [n]} \right)$ such that $v'_i(A'_i) \geq \alpha_i$ for all $ i \in [n]$.  \\ {\bf Output:} An allocation $\mathcal{A}=(A_1,\ldots ,A_n)$ such that, $v_i(A_i) \geq \alpha_i $ for all $ i \in [n]$.}
		
		\caption{Algorithm of Bouveret and Lema\^{i}tre $\ALGBL$} 
		\label{alg:bouveret}
		\begin{algorithmic}[1]
			\STATE For all $i \in [n]$ and $g \in A_i'$ set $p_g :=i $. 
			\COMMENT{This defines a sequence of agents, $P := p_1, p_2, \dots, p_m$.}
			\STATE Initialize allocation $\A=(A_1, A_2, \ldots, A_n)$ with $A_i = \emptyset$, for all $ i \in [n]$, and initialize $R \leftarrow [m]$.
			\FOR{$g =1$ to $m$}
			\STATE Pick $k \in  \argmax_{g' \in R} \left\{ v_{p_g} (g') \right\}$. 
			\STATE Update $A_{p_g}\leftarrow A_{p_g}\cup\{k\}$ and $R\leftarrow R\setminus \{k\}$.
			\ENDFOR 
			\STATE Return $\A$.
		\end{algorithmic}
	\end{algorithm}
}

\begin{theorem}
\label{thm:Bouvere1t}
Given a maximin fair division instance $I= \left( [n],[m], \{v_i \}_{i \in [n]} \right)$ and scalars $\{ \alpha_i \in \mathbb{R}\}_{i=1}^n$, let $I'=\left([n],[m],\{v'_i\}_{i \in [n]} \right)$ be the ordered instance of $I$. If there exists an allocation $\mathcal{A}'= (A'_1,A'_2,\ldots,A'_n)$ that satisfies $v'_i(A'_i) \geq \alpha_i$ for all $ i \in [n]$, then there exists an allocation $\mathcal{A}= (A_1,\ldots ,A_n)$ in which $v_i(A_i) \geq \alpha_i$, for all  $ i \in [n]$. Furthermore, given $I$ and $\mathcal{A}'$, $\ALGBL$ computes the allocation $\mathcal{A}$ in polynomial time.
\end{theorem}

\begin{proof}
Clearly, $\ALGBL$ runs in polynomial time. Now, we will show that $\ALGBL$ computes the required allocation $\A$. Let $k_g$ denote the good allocated in the $g${th} iteration of the second for-loop (Steps 3 to 6) in $\ALGBL$. Now consider agent $i$, suppose $g \in A'_i$ then $k_g \in A_i$. Note that, for any $g \neq g'$, we have $k_g \neq k_{g'}$; since, a good is removed from the set $R$ after it is allocated. Before the $g$th iteration of the second for-loop in $\ALGBL$, exactly $g-1$ goods had been allocated. Therefore $k_g$ is among the top $g$ goods for agent $i$. Hence, for all $g \in A'_i$, $v_i(k_g) \geq v'_i(g)$. This implies that, $\sum_{g \in A'_i} v_i(k_g) \geq \sum_{g \in A'_i} v'_i(g)$, i.e., $v_i(A_i) \geq v'_i(A'_i)$. Hence the condition that $v'_i(A'_i) \geq \alpha_i$ gives us $v_i(A_i) \geq \alpha_i$ for all $i$. This completes the proof.
\end{proof}

Note that the maximin share depends on the values of the goods but not on the order, hence the maximin share of an agent $i$ in an instance $I$ is equal to her maximin share in the ordered instance $I'$.  Therefore, instantiating Theorem~\ref{thm:Bouvere1t}, we have the following corollary.

\begin{corollary}
\label{thm:Bouveret}
Given a maximin fair division instance $I= \left( [n],[m], \{v_i \}_{i \in [n]} \right)$ and $\alpha \in \mathbb{R}$ write $I'$ to denote the ordered instance of $I$ and $\mu_i$ (respectively $\mu_i'$) to denote the maximin share of agent $i$ in $I$ (respectively $I'$). If there exists an  allocation $\mathcal{A}'=(A'_1,A'_2,\ldots,A'_n)$ that satisfies $v'_i(A'_i)\geq \alpha\mu'_i $, for all $ i\in [n]$, then there exists is an allocation $\mathcal{A}=(A_1,A_2, \ldots ,A_n)$ in which  $v_i(A_i)\geq\alpha\mu_i$, for all $ i \in [n]$. Furthermore, given $I$ and $\mathcal{A}'$, the allocation $\mathcal{A}$ can be computed in polynomial time.
\end{corollary}

\subsection{Envy Graph Algorithm}
As shown in Corollary~\ref{thm:Bouveret}, we only need address the setting in which the agents value the goods in the same order. Hence, in the remainder of this section, we solely focus on ordered instances. That is, we will establish the desired approximation for instances wherein the goods are indexed, say $g_1, g_2, \ldots, g_m$, such that for every agent $i$, we have $ v_{i} (g_a) \geq v_{i} (g_b)$, for all indices $a < b$. The approximation algorithm given in this section iteratively allocates the goods in decreasing order of value (i.e., in increasing order of their index) and maintains a partial allocation, $\A=(A_1, \ldots, A_n)$, of the goods assigned so far. In order to assign a good, the algorithm selects a bundle, $A_i$, by considering a directed graph, $G(\A)$, that captures the envy between agents. The nodes in this envy graph represent the agents and it contains a directed edge from $i$ to $j$ iff $i$ envies $j$, i.e., $v_i(A_i) < v_i(A_j)$. 


The following lemma was established in~\cite{lipton-envy-graph} and it shows that we can always ``resolve'' a partial allocation and obtain an acyclic \emph{envy graph}. The proof of the lemma is direct, we provide it here for completeness. 

\begin{lemma}[\cite{lipton-envy-graph}]
\label{lemma:envy}
Given a partial allocation $\A=(A_1, \ldots, A_n)$ of a subset of goods $S \subseteq [m]$, we can find another partial allocation $\B=(B_1, \ldots, B_n)$ of $S$ in polynomial time such that 
\begin{enumerate}
\item[(i)] The valuations of agents for their bundles do not decrease: $v_i(B_i) \geq v_i(A_i)$ for all $i \in [n]$.
\item[(ii)] The envy graph $G(\B)$ is acyclic.
\end{enumerate}
\end{lemma}

\begin{proof}
	If the envy graph of $\A$ is acyclic then the claim holds trivially. Otherwise, write $C=i_1 \rightarrow i_2 \rightarrow \ldots \rightarrow i_\ell \rightarrow i_1$ to denote a cycle in $G(\A)$. Now, we can reallocate the bundles as follows: for all agents $k$ not in $C$, i.e., $k \notin \{i_1, i_2, \ldots, i_\ell \}$ set $A'_k = A_k$, and for all the agents in the cycle set $A'_i$ to be the bundle of their successor in $C$, i.e., set $A'_{i_a} = A_{i_{(a+1)}}$ for $ 1\leq a < \ell$ along with $A'_{i_\ell} = A_{i_1}$.
	
	Note that after this reallocation we have $v_i(A'_i) \geq v_i(A_i)$ for all $i \in [n]$. Furthermore, the number of edges in $G(\A)$ strictly decreases: the edges in $C$ do not appear in the envy graph of $(A'_1, \ldots, A'_n)$ and if an agent $k$ starts envying an agent in the cycle, say agent $i_a$, then $k$ must have been envious of $i_{a+1}$ in $\A$. Edges between agents $k$ and $k'$ which are not in $C$ remain unchanged, and edges going out of an agent $i$ in the cycle $C$ can only get removed, since agent $i$'s valuation for the bundle assigned to her increases. Therefore, we can repeatedly remove cycles and keep reducing the number of edges in the envy graph to eventually a find a partial allocation $\B$ that satisfies the stated claim.
\end{proof}

\begin{remark}
Lemma~\ref{lemma:envy} can also be established by applying the \emph{top-trading cycle} (TTC) algorithm~\cite{shapley1974cores} among the $n$ agents and the bundles of the given partial allocation $\A$. Since the TTC algorithm (after potentially reassigning the bundles) is guaranteed to return an allocation, say $\mathcal{C}$, in which no group of agents (coalition) can improve their valuations by mutually exchanging bundles, we get that the the envy graph $G(\mathcal{C})$ is acyclic. In addition, part (i) of Lemma~\ref{lemma:envy} follows from the fact that the reallocations (trades) along the TTC will never decrease the valuations of the involved agents.
\end{remark}

Our algorithm ($\ALG$) uses Lemma~\ref{lemma:envy} and is detailed below. Lemma~\ref{lemma:efx} states that $\ALG$ returns an EFX allocation and is proved below.

{
	\begin{algorithm}
		{{\bf Input :} An ordered instance (i.e., an instance wherein the agents value the goods in the same order) with $n$ agents, $m$ indivisible goods, and valuations $v_{i,g}$ for each $i \in [n]$ and $g \in [m]$. \\ {\bf Output:}  An approximately maximin fair allocation.}
		\caption{Envy Graph Algorithm $\ALG$}
		\label{alg:envygraph}
		
		\begin{algorithmic}[1]
			\STATE Order the goods, $g_1, \ldots, g_m$, such that for every agent $i \in [n]$ and $a<b$ we have $v_{i g_a} \geq v_{i g_b}$.
			\STATE Initialize allocation $\A=(A_1, A_2, \ldots, A_n)$ with $A_i = \emptyset$ for all $ i \in [n]$.
			\FOR{$j =1$ to $m$}
			\STATE \label{step:source} Pick a vertex $i$ that has no incoming edges in the envy graph $G(\A)$, i.e., $i$ is a source vertex in $G(\A)$. \\
			\COMMENT{The algorithm maintains the invariant that $G(\A)$ is acyclic. Hence, such a vertex is guaranteed to exist.}
			\STATE Update $A_i \leftarrow A_i \cup \{ g_j \}$.
			\IF{the current envy graph $G(\A)$ contains a cycle} 
			\STATE \label{step:acyclic} Use Lemma~\ref{lemma:envy} to update $\A$ and, hence, obtain an acyclic envy graph. 
			\ENDIF
			\ENDFOR 
			\STATE Return $\A$.
		\end{algorithmic}
	\end{algorithm}
}
 
\begin{lemma}
\label{lemma:efx}
For any ordered maximin fair division instance, the partial allocation found by $\ALG$ is envy-free up to the least valued good (EFX).  
\end{lemma}

\begin{proof}
We will establish this lemma by induction. Note that the partial allocation obtained by $\ALG$ after allocating the first good $g_1$ is trivially EFX. This gives us the base case of the induction. From the induction hypothesis, we get that the partial allocation obtained by $\ALG$ at the end of the ${(j-1)}$th iteration is EFX. We will now show that the partial allocation computed by $\ALG$ after allocating the good $g_j$ is EFX as well. 

Write $\A=(A_1, A_2, \ldots, A_n)$ to denote the partial allocation of $\ALG$ at the beginning of the $j$th iteration and let $i$ be the source vertex selected in the $j$th iteration (see Step~\ref{step:source}). $\ALG$ adds good $g_j$ to the bundle of agent $i$ and the bundles of all other agents remain unchanged. Let $\A'=(A'_1, A'_2, \ldots, A'_n)$ be the partial allocation obtained after the assignment of $g_j$ and before the application of Lemma~\ref{lemma:envy} (in Step~\ref{step:acyclic}). Here, $A'_a = A_a$ for all $a \neq i$ and $A'_i = A_i \cup \{g_j\}$.

Note that any new envy that is created by assigning $g_j$ (to $i$) must be directed towards $i$: for any two agents $a, b \neq i$, if we had $v_a(A_a) \geq v_a(A_b)$ before the allocation of $g_j$, then this inequality continues to hold after the allocation (since $\ALG$ does not alter bundles $A_a$ and $A_b$), i.e., we have $v_a(A'_a) \geq v_a(A_b')$. 

Since $i$ is a source vertex in $G(\A)$, with respect to allocation $\A$ agent $i$ was not envied by any other agent. Furthermore, for every agent, $g_j$ is the least valued good among the goods considered so far (i.e., among $\{g_1, g_2, \ldots, g_{j-1}, g_j\}$).\footnote{Recall that the given instance is ordered.} These observation imply that $\A'$ is an EFX allocation as well. In particular, if for an agent $a$ we have $v_a(A'_a) < v_a(A'_i)$, then using the fact that $v_a(A'_a) = v_a(A_a) \geq v_a(A_i) = v_a(A'_i \setminus \{g_j\})$ we get that the envy between $a$ and $i$ can be resolved by removing the least valued good $g_j$.  For other agents $a, b \neq i$, the induction hypothesis ensures that (analogous to $\A$), the envy in allocation $\A'$ is up to the least valued good (recall that $A'_a = A_a$ and $A'_b = A_b$).  
 
Finally, $\ALG$ applies Lemma~\ref{lemma:envy} to $\A'$ to obtain an acyclic envy graph. We know that in this process the valuations of the agents for their bundles do not decrease. Indeed, the update performed in Lemma~\ref{lemma:envy} simply reallocates the bundles and, hence, the collection of bundles remains unchanged. Therefore, an application Lemma~\ref{lemma:envy} maintains the EFX property.

Overall, by induction, we get that all the partial allocations computed by $\ALG$ are EFX and the claim follows. 
\end{proof}

\subsection{Majorization}
\label{section:majorization}

Next we introduce the concept of \emph{majorization} and prove a technical proposition which will be useful in establishing the desired approximation guarantee. Majorization corresponds to a partial order on multisets/vectors of real number and this notion provides a framework for establishing inequalities in many settings; see~\cite{olkin2016inequalities} for an excellent exposition. Under this partial order, a multiset $X$ is related to (i.e., majorizes) another multiset $Y$ iff the prefix sums of the sorted version of $X$ are at least as large as the prefix sums of the sorted version of $Y$. Formally,  

\begin{definition}[Majorization]
A multiset $X = \{ x_i \in \mathbb{R} \mid 1 \leq i \leq n \}$ is said to majorizes another multiset $Y=\{ y_i \in \mathbb{R} \mid 1 \leq i \leq n  \}$ iff 
\begin{align*}
\sum_{i=1}^k x_{(i)} & \geq \sum_{i=1}^k y_{(i)} \qquad \text{ for all } 1 \leq k < n \quad \textrm{and} \\
\sum_{i=1}^n x_{(i)} & = \sum_{i=1}^n y_{(i)}.
\end{align*}
Here $x_{(i)}$ ($y_{(i)}$) is the $i$th largest element of $X$ ($Y$).\footnote{Ties for the $i$th largest position are broken arbitrarily.} 
\end{definition}


A proof of Proposition~\ref{prop:maj} can be found in~\cite{lpt-wu}. For completeness, we provide one here as well. 
\begin{proposition}
\label{prop:maj}
Let $v, v' \in \mathbb{R}$ be two elements of a multiset $A$ of real numbers. In addition, say $u,u' \in \mathbb{R}$ satisfy $u+u' = v + v'$ and $|u - u'| < |v - v'| $. Then $A$ majorizes $ \left( A \setminus \{v,v'\}\right) \cup \{u, u'\}$.
\end{proposition}
\begin{proof}
	Define $f_t(S)$ to be the sum of the largest $t$ elements in $S$. Consider the set $A=\{a_1,a_2,\dots,a_n\}$ such that $a_1\geq a_2\geq \dots \geq a_n$. Consider $ A'=\left( A \setminus \{v,v'\}\right) \cup \{u, u'\}$ such that $u+u' = v + v'$ and $|u - u'| < |v - v'| $. Without loss of generality, let $v=a_{i_1} \geq v'=a_{i_2}$ and let $u \geq u'$. Since $u+u' = v + v'$ and $|u - u'| < |v - v'| $, we have $v \geq u \geq u' \geq v'$. Therefore, the non-increasing sequence of elements in $A'$ is
	$$a_1 \geq a_2 \geq \dots \geq a_{i_1-1} \geq a_{i_1+1} \geq \dots \geq u \geq a_{i_3} \geq \dots \geq u' \geq a_{i_4} \geq \dots \geq a_{i_2-1} \geq a_{i_2 +1} \geq \dots \geq a_n  $$
	
	We now do case-wise analysis:
	\begin{enumerate}
		\item For $t<i_1$ and $t\geq  i_2$: $f_t(A)=f_t(A')$.
		\item For $i_1 \leq t < i_3 - 1$: $f_t(A')=f_t(A)-a_{i_1}+a_{t+1} \leq f_t(A)$.
		\item For $i_3-1 \leq t < i_4$: $f_t(A')=f_t(A)-a_{i_1}+u \leq f_t(A)$.
		\item For $i_4 \leq t < i_2$: $f_t(A')=f_t(A)-a_{i_1}-a_t+u+u'=f_t(A)+a_{i_2}-a_t \leq f_t(A)$. 
	\end{enumerate}
	
	Therefore, $f_t(A')\leq f_t(A)$ for any $1\leq t \leq n$, i.e., $A$ majorizes $ \left( A \setminus \{v,v'\}\right) \cup \{u, u'\}$.
\end{proof}

\subsection{Proof of Theorem~\ref{thm:main}}
To establish the stated claim, in the remainder of the proof we will show that $v_1(A_1) \geq \frac{2n}{3n-1} \mu_1$; an analogous argument establishes the desired bound, $v_i(A_i) \geq \frac{2n}{3n-1} \mu_i$, for all agents $i \in [n]$.

Write $\A=(A_1, \ldots, A_n)$ to denote the allocation returned by the algorithm $\ALG$. Consider the set of goods that have value more than $\frac{1}{2} v_1(A_1)$; specifically, define $\tau := \argmin \{ j \mid v_1(g_j) \leq \frac{1}{2} v_1(A_1)\}$ and let $H$ denote the set of high valued goods $H:= \{g_1, g_2, \ldots, g_{\tau-1} \}$. In addition, write $\Pa=(P_1, P_2, \ldots, P_n)$ to denote the partial allocation that $\ALG$ computes for $H$.\footnote{Without loss of generality, we can assume that $G(\Pa)$ is acyclic, since an application of Lemma~\ref{lemma:envy} in $\ALG$ simply reassigns the bundles between agents and the constituent bundles in $\Pa$ do not change.}

Also, let $t$ denote the smallest iteration count at which $\ALG$ assigns a good (in particular, good $g_{t}$) to a bundle of size two. Hence, every bundle in the partial allocation of the first $t-1$ goods is of size at most two.

Note that $\ALG$ keeps updating the partial allocations by adding goods to existing bundles.\footnote{Though, the bundle assigned to an agent might change during an application of Lemma~\ref{lemma:envy}.} This observation implies that the cardinalities of the bundles in successive partial allocations keep on increasing. This monotonic growth of the bundles also ensures that for each set $A_i$ (in the final allocation $\A$), there exists a unique set $P_{j}$ (in the partial allocation $\Pa$) such that $P_{j} \subseteq A_i$. Moreover, Lemma~\ref{lemma:envy} guarantees that as $\ALG$ progresses the valuation of the agents also does not decrease; in particular, $v_1(A_1) \geq v_1(P_1)$.

Let $\Q=(Q_1, Q_2, \ldots, Q_n)$ be the partial allocation considered by $\ALG$ at the beginning of the $t$th iteration. The definition of $t$ ensures that $|Q_i| \leq 2$ for all $i \in [n]$. Below we will show that $\tau \leq t$ and, therefore, using the observation that the cardinalities of the bundles are nondecreasing, we get the following bound: $|P_i| \leq 2$ for all $i \in [n]$. This, in turn, leads to the inequality $\tau \leq 2n$. 

Recall that good $g_t$ was the first good that was assigned to a set of cardinality two. Write $Q_b=\{ g_x, g_y \}$ to denote this set. Since $g_t$ is included in $Q_b$, $b$ must have been a source of the acyclic envy graph $G(\Q)$. In other words, agent $1$ does not envy $b$'s current bundle, $v_1(Q_1) \geq v_1(Q_b) = v_1(g_x) + v_1(g_y)$.
Moreover, since $\ALG$ assigns goods in decreasing order of value,  $g_t$'s value is no more than that of $g_x$ and $g_y$: $v_1(g_x) \geq v_1 (g_t)$ and $v_1(g_y) \geq v_1(g_t)$. Therefore, $v_1(Q_1) \geq 2 v_1 (g_t)$. Since the valuation of agent $1$ does not decrease during the execution of $\ALG$, we get $v_1(A_1) \geq 2 v_1 (g_t)$. By definition, $g_\tau$ is the smallest indexed good that satisfies this inequality and, hence, the inequality $\tau \leq t$ holds.



We will first address the bundles that contain a good with index greater than $\tau-1$. In particular, 
we establish the following claim.
\begin{claim}
\label{claim:propo}
In the partial allocation retuned by $\ALG$, $\A=(A_1, \ldots, A_n)$, the following inequality holds for all bundles $A_i$ that satisfy $A_i \cap \{g_\tau, g_{\tau+1}, \ldots, g_m \} \neq \emptyset$
\begin{align}
\label{ineq:propo}
v_1(A_1) \geq \frac{2}{3} v_1(A_i)
\end{align}
\end{claim}

\begin{proof}
Consider a bundle $A_i$ which contains a good $g_a$ with index $a \geq \tau$. Lemma~\ref{lemma:efx} implies that $v_1(A_1) \geq v_1(A_i) - v_1(g_a)$. Also, note that $v_1(A_1) \geq 2 v_1(g_\tau) \geq 2 v_1(g_a)$; here, the second inequality follows from the fact that the values of the goods are ordered. Therefore, the claim holds.
\end{proof} 

We now need to argue about bundles in $\A$ that do not contain any good with index greater than $\tau-1$. For each such bundle, $A_i$, we have a unique set $P_j$ in the partial allocation $\Pa$ such that $A_i = P_j$. This follows from the fact that after the $(\tau-1)$th iteration no good is allocated to any of these bundles. Therefore, to get a handle on bundles in $\A$ that do not contain any good with index greater than $\tau-1$ we will consider the partial allocation $\Pa$. By definition, $\Pa$ is a partial allocation of $H =\{g_1, \ldots, g_{\tau-1} \}$, i.e., $\cup_{i=1}^n P_i = H$. As mentioned above the cardinality of $H$ (i.e., $\tau -1$) is at most $2n$. Write $h := \max\{ 0 , \tau - 1 - n \}$. We will consider the case in which $\tau \geq n+1$---and, hence $\tau - 1  = n + h$---the other case wherein $\tau < n+1$ is simpler and follows analogously.

\begin{claim}
\label{claim:H-ord}
The partial allocation $\Pa$ consists of (up to reordering) the following bundles: 
\begin{align}
\{g_1\}, \{g_2\}, \ldots, \{g_{n-h} \}, \{ g_{n-h+1}, g_{n+h} \}, \{g_{n-h+2}, g_{n+h -1} \}, \ldots, \{g_{n-1}, g_{n+2} \}, \{g_n, g_{n+1} \}.
\end{align}
\end{claim}
\begin{proof}
	In the first $n$ iterations $\ALG$ allocates the goods $\{g_1, g_2, \ldots, g_n\}$ to $n$ distinct bundles. At this point of time---since the values of the goods are ordered---the agent with bundle $\{g_n\}$ envies all the other agents.\footnote{Without loss of generality, we can assume that the valuations are distinct.} Hence, good $g_{n+1}$ is allocated to the bundle $\{g_n\}$. By definition, till the $t$th iteration (with $t \geq \tau -1 = n+h$) no bundle gets more than two items. Therefore, it must be the case that all the remaining $h$ goods in $H$ are allocated to bundles of size one. Now, we can argue inductively that when allocating good $g_{n+k}$ (for $1 < k \leq h$), $\ALG$ has the following singleton bundles: $\{g_1\}, \{g_2\}, \ldots, \{g_{n-k}\}, \{g_{n-k+1}\}$ and, among the agents that have these bundles, only the agent with bundle $\{g_{n-k+1}\}$ can be a source in the envy graph.\footnote{Recall that the values of the goods are ordered.} Therefore, good $g_{n+k}$ gets assigned to bundle $\{g_{n-k+1}\}$. Overall, we get that $\Pa$  consists of the bundles mentioned above. 
\end{proof}


\begin{figure}
\begin{center}
	\includegraphics[scale=1.1]{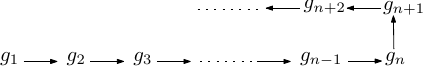}
	\caption{Sequence of allocation of the first $\tau$ goods.}
\end{center}
\end{figure}



In the remainder of the proof we say that a partial allocation $\R=(R_1, \ldots, R_n)$ majorizes another partial allocation $\R'=(R'_1, \ldots, R'_n)$ iff $\{ v_1(R_i) \}_i$ majorizes $\{ v_1(R'_i)\}_i$. Specifically, we have the following claim for $\Pa$. The proof of Claim~\ref{claim:maj} is similar to the proof of Lemma 4 in~\cite{lpt-wu}.
\begin{claim}
	\label{claim:maj}
	Every partial allocation of $H$ (i.e., every $\R \in \Pi_n(H)$) majorizes  $\Pa$.
\end{claim}
\begin{proof}
	We prove Claim~\ref{claim:maj} by a repeated application of Proposition~\ref{prop:maj}. From Claim~\ref{claim:H-ord}, we know that every bundle in $\Pa$ is of size at most 2. In $\Pa$, the goods $\{g_1,g_2,\dots,g_{n-h}\}$ are not paired up with any other goods in $H$. And in $\Pa$, the good $g_{n-i}$ is paired up with the good $g_{n+1+i}$ for all $0\leq i \leq h-1$.
	For some $t$, let $\R \in\Pi_n(H)$ be a partition that minimizes the sum of the value of the largest $t$ bundles in $\R$ (with respect to agent~1). We will show that $\R$ can transformed to $\Pa$ without increasing the total weight of the largest $t$ bundles:
	\begin{enumerate}
		\item For any empty bundle in $\R$, we can move a good from any bundle containing more than one good to this empty bundle. Note that, this step does not increase the weight of the sum of the largest $t$ bundles in $\R$. i.e. We can modify $\R$ so that no bundle in $\R$ is empty.
		\item Consider any $i\leq n-h$ such that $g_i$ is part of a bundle with more than 1 good in $\R$. By pigeonhole principle, there exists a $j>n-h$ such that $\{g_j\}$ is a singleton bundle in $\R$. Then from Proposition~\ref{prop:maj}, by exchanging $g_i$ and $g_j$, we will not increase the total weight of the largest $t$ bundles.
		i.e., We can modify $\R$ so that, for all $i\leq n-h$, $\{g_i\}$ is a singleton bundle. 
		\item Consider any bundle in $\R$ with at least 3 goods (We use $R_i$ to denote such a bundle). Let $g=\min\{R_i\}$ be the smallest good in $R_i$. Then by pigeonhole principle, there exists a $j>n-h$ such that $\{g_j\}$ is a singleton bundle in $R$. 
		$$(v_1(R_i)-v_1(\{g\})) - (v_1(\{g\}) + v_1(\{g_j\})) < v_1(R_i) - v_1(\{g_j\}) $$
		Since $j>n-h$, this implies that $v_1(\{g_j\}) \leq v_1(\{g_{n-h+1}\})$. Using the EFX property followed by the definition of $H$, we have $v_1(\{g_{n-h+1}\})\leq v_1(A_1) \leq 2v_1(\{g_{n+h}\})$. Since $R_i$ has at least 3 items, we know that $2v_1(\{g_{n+h}\}) \leq 2v_1(\{g\}) \leq v_1(R_i)-v_1(\{g\})$.
		i.e., $ v_1(\{g_j\}) \leq v_1(\{g_{n-h+1}\}) \leq v_1(A_1) \leq 2v_1(\{g_{n+h}\}) \leq 2v_1(\{g\}) \leq v_1(R_i)-v_1(\{g\}) $.
		
		Hence, we have:
		$$ (v_1(R_i)-v_1(\{g\})) - (v_1(\{g\}) + v_1(\{g_j\})) \geq -v_1(\{g\}) \geq  v_1(\{g_j\})-v_1(R_i)$$
		Therefore, $|(v_1(R_i)-v_1(\{g\})) - (v_1(\{g\}) + v_1(\{g_j\}))| < v_1(R_i) - v_1(\{g_j\})$.
		Then from Proposition~\ref{prop:maj}, by moving $g$ from $R_i$ to the bundle $\{g_j\}$, we will not increase the total weight of the largest $t$ bundles.
		i.e., We can adjust $\R$ so that every bundle has atmost 2 items.
		\item  Therefore, now the only bundles that are in $\R$ but not in $\Pa$ are bundles having 2 items. Suppose in $\R$, there are 2 bundles $\{g_{i_1},g_{i_2}\}$ and $\{g_{j_1},g_{j_2}\}$, such that $i_1<j_1<i_2<j_2$. This implies that, $|v_1(\{g_{i_1},g_{i_2}\})-v_1(\{g_{j_1},g_{j_2}\})| \leq |v_1(\{g_{i_1},g_{j_2}\})-v_1(\{g_{j_1},g_{i_2}\})|$. Therefore from Proposition~\ref{prop:maj}, by exchanging $g_{i_2}$ and $g_{j_2}$, we will not increase the total weight of the largest $t$ bundles.
		i.e., We can adjust $\R$ so that such pairs of bundles are not their in $\R$. Now, $\R$ is the same as $\Pa$ (up to a reordering). 
	\end{enumerate}
	
	Therefore we have showed that, for any $t$, $\Pa$ minimizes the sum of the largest $t$ bundles.
	
\end{proof}

We will now complete the proof of Theorem~\ref{thm:main}. Let $\ell$ be the number of bundles in $\{A_2,A_3,\dots,A_n\}$ that do not contain any good with index greater than $\tau-1$. As mentioned above, for every such bundle, $A_i$ there exists a unique bundle $P_j$ in $\Pa$ such that $A_i = P_j$. 
Say, by reindexing, that $P_{n-\ell+1}, P_{n-\ell+2}, \ldots, P_n$ are these $\ell$ bundles in $\A$ that do not satisfy the condition in Claim~\ref{claim:propo}. 

Let $(M_1, M_2, \ldots, M_n)$ denote a partition of the $[m]$ goods that achieves the maximin share for player $1$, i.e., $\min_{j \in [n]} v_1(M_j) = \mu_1$ (see Definition~\ref{defn:maxmin}). Consider the partial allocation $(M_1 \cap H, M_2 \cap H, \ldots, M_n \cap H )$ and index the sets $M_i$s such that $v_1(M_1 \cap H) \leq v_1(M_2 \cap H) \leq \ldots \leq v_1(M_n \cap H)$. Claim~\ref{claim:maj} implies that $(M_1 \cap H, M_2 \cap H, \ldots, M_n \cap H )$ majorizes $\Pa$. Hence, $\sum_{i=n-\ell+1}^{n} v_1(M_i \cap H) \geq \sum_{i=n-\ell+1}^{n} v_1(P_i).$ This inequality along with the fact that the valuations are monotone lead to the  following useful bound
\begin{align}
\label{ineq:prefix}
\sum_{i=n-\ell+1}^{n} v_1(M_i) \geq \sum_{i=n-\ell+1}^{n} v_1(P_i).
\end{align}

Recall that $P_{n-\ell+1}, P_{n-\ell+2}, \ldots, P_n$ are bundles in $\A$, and the remaining $(n-\ell)$ bundles of $\A$---say, $A_1, A_2, \ldots, A_{n-\ell}$---satisfy the inequality (\ref{ineq:propo}). Since $v_1$ is additive, we have $\sum_{i=1}^n v_1(M_i) = \sum_{i=1}^n v_1(A_i) $. Therefore, inequality (\ref{ineq:prefix}) provides the following bound
\begin{align}
\sum_{i=1}^{n-\ell} v_i(M_i)  & \leq \sum_{i=1}^{n-\ell}  v_1(A_i) \\  & \leq v_1(A_1)+(n - \ell-1) \frac{3}{2} v_1 (A_1) \qquad \text{   (via Claim~\ref{ineq:propo})} \label{ineq:final}
\end{align}

Overall, inequality (\ref{ineq:final}) implies that $v_1(A_1)$ is at least $\frac{2(n-\ell)}{3(n-\ell)-1}$ times the average of the $(n-\ell)$ sets $M_{\ell+1}, M_{\ell+2}, \ldots, M_n$. Hence, $v_1(A_1) \geq \frac{2n}{3n-1} \min_{j \in [n]} v_1(M_j) = \frac{2n}{3n-1} \mu_1$. This completes the proof.

\section{Submodular Valuations}

In this section we show that when the agents have submodular valuations, then a $0.21$-approximate maximin fair allocation is guaranteed to exist and, moreover, it can be found in polynomial time. Our results only need oracle access to the submodular functions.

\subsection{Finding an Approximate Maximin Fair Allocation}

We compute an approximate maximin fair allocation by employing Algorithm~\ref{alg:mms-const} (\ALGS) as a subroutine. Recall that computing the maximin shares is an {\rm NP}-hard problem, hence we execute the subroutine $\ALGS$ with threshold $\tau_i$ (for each agent $i \in [n]$), instead of the actual maximin share $\mu_i$. Below we detail properties of $\ALGS$ which enable us to update $\tau_i$s and successively use this subroutine to obtain the desired approximation guarantee. 

$\ALGS$ is quite direct: it takes as input thresholds, $\tau_i$s, allocates high-valued (with respect to $\tau_i$) goods as singleton bundles, and then partitions the remaining goods in a round-robin fashion. The technical contribution here is to show that, for each agent $i$, if $\tau_i$ is less than or equal to the maximin share $\mu_i$, the bundle $P_i$ allocated to $i$ by $\ALGS$ satisfies $v_i(P_i) \geq \frac{1}{3} \left( 1 - \frac{1}{e} \right) \tau_i$. It is relevant to note that this guarantee holds independently for each agent $i$ (as long as $\tau_i \leq \mu_i$) and will not be violated even if $\tau_j > \mu_j$ for $j \neq i$. Formally, we establish this guarantee in Lemma~\ref{lemma:ind-guarantee} and use its contrapositive version (i.e., if $\ALGS$ returns a $P_i$ such that $v_i(P_i) < \frac{1}{3} \left( 1 - \frac{1}{e} \right) \tau_i$, then it must be the case that $\tau_i > \mu_i$) to establish the main result of this section (Theorem~\ref{thm:mms-const}). 

In particular, Algorithm~\ref{alg:mms-final} ($\ALGF$) starts by setting thresholds $\tau_i$s to be more than the maximin shares. Then, it (geometrically) decreases the threshold for agent $i$, if the partition returned by \ALGS \ does not satisfy $v_i(P_i) \geq \frac{1}{3} \left( 1 - \frac{1}{e} \right)  \tau_i$. In such a case, as stated previously, $\tau_i$ must have been greater than $\mu_i$ and, hence, decreasing the threshold is justified. Overall, in Section~\ref{sect:proof-submod} we show that \ALGF \ efficiently finds thresholds, for each agent, which are comparable to their maximin share, and establish the following theorem. 

\begin{theorem} [Main Result for Submodular Valuations] \label{thm:mms-const}
Given a maximin fair division instance with $m$ indivisible goods and $n$ agents that have nonnegative, monotone, submodular valuations $v_i: 2^{[m]} \mapsto \mathbb{R}_+$, $1 \leq i  \leq n$, the algorithm $\ALGF$ finds, in polynomial time, an allocation $\mathcal{P}=(P_1,P_2,\ldots, P_n)$ which satisfies $v_i(P_i) \geq 0.21 \ \mu_i$, for all  $ i \in [n]$. Here $\mu_i$ is the maximin share of agent $i$.
\end{theorem}

We begin by stating the subroutine \ALGS \ and establishing its key property in Lemma~\ref{lemma:ind-guarantee}.\footnote{Note that this lemma does not rule out the possibility that \ALGS \ finds a bundle $P_i$ which satisfies $v_i(P_i) \geq 0.21 \tau_i$, even if $\tau_i > \mu_i$.}

{
	\begin{algorithm}[H]
		{
		{\bf Input:} An instance over $m$ indivisible goods and $n$ agents, whose valuations, $v_i: 2^{[m]} \mapsto \mathbb{R}_+$, are nonnegative, monotone and submodular, along with thresholds $\tau_1, \ldots, \tau_n$. \\ 
		{\bf Output:} An allocation $\mathcal{P}=(P_1,P_2, \dots ,P_n)$ such that $v_i(P_i) \geq \frac{1}{3} \left(1 - \frac{1}{e}\right) \tau_i \geq 0.21 \ \tau_i$, for each agent $ i \in [n]$ whose  threshold satisfies $\tau_i \leq \mu_i$. Here $\mu_i$ is the $n$-maximin share of agent $i$.  
		\caption{Computation of allocations with respect to thresholds $\ALGS$}
		\label{alg:mms-const}
		\begin{algorithmic}[1]
			\STATE Initialize set of agents $A=[n]$, set of goods  $G=[m]$, and $P_i = \emptyset$ for all $i \in A$.
			\WHILE{there exist agent $i\in A$ and good $g \in G$ such that $v_i(g)\geq  \frac{1}{3} \left(1 - \frac{1}{e}\right) \ \tau_i$}
			\STATE Allocate $P_i\leftarrow\{g \}$, and update $A \leftarrow A \setminus\{i\}$ and $ G \leftarrow G\setminus\{g\}$.
			\ENDWHILE
			\STATE Assume, via reindexing, that the set of remaining agents $A=\{1,2, \ldots, |A|\}$. 
			\WHILE{$G\neq\emptyset$}
			\FOR{$i=1$ to $|A|$}
			\STATE Pick $g\in \argmax_{g' \in G} \ \left\{v_i(P_i\cup \{g' \})-v_i(P_i) \right\}$.
			\STATE $P_i\leftarrow P_i\cup\{g\}$, $G\leftarrow G\setminus\{g\}$
			\ENDFOR
			\ENDWHILE			
			\RETURN partition $\mathcal{P}=(P_1, P_2, \ldots, P_n)$.
		\end{algorithmic}
	}
	\end{algorithm}
}

\begin{lemma} \label{lemma:ind-guarantee}
For a given a maximin fair division instance with $m$ indivisible goods and $n$ agents, whose valuations, $v_i: 2^{[m]}\mapsto \mathbb{R}_+$, are nonnegative, monotone and submodular, let $\mathcal{P}=(P_1,P_2, \dots ,P_n)$ be the allocation returned by $\ALGS$ with thresholds $\tau_i \in \mathbb{R}_+$, $1 \leq i \leq n$. Then, for each agent $i$ whose input threshold satisfies $\tau_i \leq \mu_i$, we have 
\begin{align*}
v_i(P_i) \geq \frac{1}{3}\left( 1 - \frac{1}{e}\right) \tau_i \geq 0.21 \ \tau_i.
\end{align*} 
Here, $\mu_i$ is the $n$-maximin share of agent $i$.
\end{lemma}

	


{
	\begin{algorithm}
		{
			{\bf Input:} An instance with $m$ indivisible goods and $n$ agents, whose valuations, $v_i: 2^{[m]} \mapsto \mathbb{R}_+$, are nonnegative, monotone and submodular. \\ 
			{\bf Output:} An allocation $\mathcal{P} = (P_1,P_2, \ldots, P_n)$ such that, $v_i (P_i) \geq 0.21 \mu_i$ for all $ i \in [n]$. Here $\mu_i$ is the $n$-maximin share of agent $i$. 
			\caption{Computation of an approximate maximin fair allocations $\ALGF$}
			\label{alg:mms-final}
			\begin{algorithmic}[1]
			\STATE For all $i \in [n]$, initialize $\tau_i  = v_i([m])$ and $P_i = \emptyset$.
			\COMMENT{We assume that $\mu_i >0$ for each agent $i$; agents that do not satisfy this assumption can be simply detected and removed from consideration.}
				\STATE Initialize $V = \{ i \in [n] \mid v_i (P_i) < 0.21 \tau_i \}$ (i.e., $V =  [n]$) and set $\delta \in (0,1)$ to be an arbitrarily small constant. 
			\WHILE{ $V \neq \emptyset$ }
				\STATE For all $i \in V$, update $\tau_i \leftarrow \frac{1}{(1+ \delta)} \tau_i $.
 			\STATE  Update the partition by executing \ALGS \ with the current threshold values: $(P_1, P_2, \ldots, P_n) \leftarrow \ALGS(\tau_1, \ldots, \tau_n)$.
			\STATE Update $V \leftarrow \{ i \in [n] \mid v_i (P_i) < 0.21 \tau_i \}$
				\ENDWHILE		
			\RETURN Partition $\mathcal{P}=(P_1, P_2, \ldots, P_{n}) $
			\end{algorithmic}
		}
	\end{algorithm}
}

\subsection{Multilinear Extensions} \label{sec:Multilinear}

Our analysis of \ALGS \ rests on studying the \emph{multilinear extension} of the valuation functions. This concept has been used in recent results for constrained submodular maximization; see, e.g.~\cite{calinescu2011maximizing,vondrak-matroid,kulik2009maximizing,lee2010maximizing,vondrak2013symmetry, chekuri2010dependent, chekuri2014submodular}. Formally, 

\begin{definition} [Multilinear Extension] \label{def:multilinear-extension}
For a function $f: 2^{[m]} \mapsto \mathbb{R}$, the multilinear extension $F: [0,1]^{[m]} \mapsto \mathbb{R}$ is defined as follows:
\begin{align*}
F(x) & :=\E_{R \sim x} [f(R)] = \sum_{R \subseteq [m] } f(R) \ \left( \prod_{g \in R} x_g  \  \prod_{g \in ([m] \setminus R)} (1-x_g) \right).
\end{align*}
Here sampling from $x \in [0,1]^m$ corresponds to selecting a random subset $R \subseteq [m]$ in which each $g \in [m]$ appears independently with probability $x_g$.
\end{definition}

We say an $n$-tuple $\chi :=(x^1,x^2,\dots,x^n)$ with $x^i\in[0,1]^m$ is a \emph{fractional allocation} iff $\sum_{i=1}^n x^i_g \leq 1$ for all $ g \in [m]$. A binary fractional allocation corresponds to a (partial) allocation. Also, note that the set of all fractional allocations forms a partition matroid polytope over the set $[m]$. 
Next we state and prove a useful property of the \emph{uniform} fractional allocation.

\begin{lemma}[Proportionality] \label{lem:prop}
Let $v_i: 2^{[m]} \mapsto \mathbb{R}_+$ denote nonnegative, monotone, submodular valuations of agents $i \in [n]$ over $m$ indivisible goods. Then, the fractional allocation $ \omega =(u^1,u^2,\dots,u^n)$---in which $u^i \coloneqq \left( \frac{1}{n}, \frac{1}{n}, \ldots,\frac{1}{n} \right) \in [0,1]^m$ for all $i \in [n]$---satisfies $V_i(u^i)  \geq \left(1-\frac{1}{e} \right)\mu_i $ for all $i \in [n]$.
\noindent
Here $V_i:[0,1]^m\mapsto \mathbb{R}_+$ is the multilinear extension of $v_i$ and $\mu_i$ is the maximin share of agent $i$.
\end{lemma}
\begin{proof}
	Let $f: 2^{[m]}\mapsto \mathbb{R}_+$ be a nonnegative, monotone, submodular valuation function. Furthermore, for a fractional allocation $\chi =(x^1,x^2, \ldots, x^n)$, write $\mathcal{F}(\chi)$ to denote expected social welfare of $n$ agents with identical valuation function $f$, i.e., 
	\begin{align*}
	\mathcal{F}(\chi) &  := \sum_{i=1}^n F(x^i) \\ & = \sum_{i=1}^n \E_{R_i \sim x^i} [f(R_i) ].
	\end{align*}
	Here, $x^i \in [0,1]^m$ and $F$ is the multilinear extension of $f$. 
	
	Vondrak \cite{vondrak-matroid} (see Remark 2 on page 6) established that under the uniform fractional allocation $\omega=(u^1,u^2,\dots,u^n)$ the expected social welfare is at least $(1 - e^{-1})$ times the optimal:
	\begin{align}
	\mathcal{F}( \omega ) \geq \left( 1- \frac{1}{e}\right) \ \ \ \max_{(P_1, \ldots, P_n) \in \Pi_n(m)}  \ \sum_{i=1}^n f(P_i) \label{ineq:vondrak}
	\end{align}
	
	Write $(M_1^*, \ldots, M_n^*)$ to denote a partition that achieves the maximin share with respect to $f$, i.e., $(M_1^*, \ldots, M_n^*) \in \argmax_{(P_1, \ldots, P_n) \in \Pi_n(m)} \ \min_{i \in [n]} f(P_i)$. In addition, let $\mu$ be the maximin share; hence, $ \mu = \min_{i \in [n]} f(M_i^*)$
	
	Using the fact that $u^i = u^j $ for all $i, j \in [n]$ we have the following bound for the multilinear extension $F$.
	\begin{align*}
	F(u^i) & = \frac{1}{n} \sum_{j=1}^n F(u^j) \\
	& = \frac{1}{n} \mathcal{F}(\omega)  \\
	& \geq  \frac{1}{n} \left( 1- \frac{1}{e}\right) \ \ \ \max_{(P_1, \ldots, P_n) \in \Pi_n(m)}  \ \sum_{i=1}^n f(P_i)  \qquad \text{ (via inequality (\ref{ineq:vondrak})) } \\
	& \geq  \frac{1}{n} \left( 1- \frac{1}{e}\right) \ \  \sum_{i=1}^n f(M_i^*) \\ 
	& \geq   \frac{1}{n} \left( 1- \frac{1}{e}\right) \ \  n \mu
	\end{align*}
	
	
	Therefore, we have $F(u^i) \geq  \left( 1- \frac{1}{e}\right)  \mu$. We can simply instantiate this inequality for each valuation function $v_i$ to obtain the stated claim that $V_i(u^i) \geq (1-\frac{1}{e}) \mu_i$, for all $i \in [n]$.
\end{proof}

Recall that $f_H$ denotes the marginal function with respect to subset $H \subseteq [m]$, $f_H(S):= f(H\cup S)- f(H)$ for all $ S\subseteq [m]$. Write $F_H$ to denote the multilinear extension of $f_H$. 
	
We define $x(S)$ to be the projection of a vector $x  \in [0,1]^m$ onto $S \subseteq[m]$. That is, the $j$th component of $x(S)$ is equal to $x_j$ if $j \in S$ and it is equal to zero otherwise. 

In addition, we will say that a vector $y$ is supported over a set $S$ iff the set of nonzero components of $y$ is equal to $S$.  Hence, for $u \coloneqq \left( \frac{1}{n},\frac{1}{n}, \ldots, \frac{1}{n} \right)\in [0,1]^m$ and $S \subseteq [m]$, the projection vector $u(S) \in [0,1]^m$ is supported over $S$. In addition, the probability of drawing a subset $R$ from $u(S)$ is nonzero iff $R \subseteq S$.

The following property of multilinear extensions will be used in the proof of Lemma~\ref{lemma:ind-guarantee}. 

\begin{proposition}  \label{claim:main-ineq2}
Given a monotone, submodular function $f: 2^{[m]} \mapsto \mathbb{R}_+$ along with subset  $P \subseteq [m]$, good $g \in [m]$, and vector $x \in [0,1]^m$, the multilinear extensions of the marginal functions satisfy $F_{P \cup \{g\}} (x )  \geq F_{P} (x ) - f_P(g)$.
\end{proposition}
\begin{proof}
	For any $R\subseteq [m]$, by definition of the marginal function,  we have 
	\begin{align*}
	f_{P\cup\{g\}}(R) & = f(R\cup P\cup\{g\}) - f(P\cup\{g\}) \\
	& \geq f(R\cup P) - f(P\cup\{g\}) \qquad \text{ (monotonicity of $f$) }\\
	& = [f(R\cup P)-f(P)] - [f(P\cup\{g\})-f(P)]\\
	& = f_P(R) - f_P(g)
	\end{align*}
	Taking expectation with respect to $x$ (i.e., drawing $R \sim x$) we get the stated claim $F_{P\cup\{g\}}( x  ) \geq F_P(x) - f_P(g) $.
\end{proof}

\subsection{Proof of Lemma~\ref{lemma:ind-guarantee}}
\label{sect:proof-ind-guarantee}

In the proof we will fix an agent $i$ and---under the assumption that $\tau_i \leq \mu_i$---establish the stated claim, $v_i(P_i) \geq \frac{1}{3}\left(1 - \frac{1}{e} \right) \tau_i$, for $i$; an analogous argument establishes the claim for all other agents.
	
If agent $i$ was assigned a good in the first while-loop of $\ALGS$ (Algorithm~\ref{alg:mms-const}), then we have the desired inequality $v_i(P_i) \geq \frac{1}{3}\left(1 - \frac{1}{e} \right) \tau_i$. Otherwise, the first while-loop of $\ALGS$ terminates with, say, the good set $G=\{1,2, \ldots, |G|\}$ (denoting the set of unallocated goods) and agent set $A=\{1, 2, \ldots, |A|\}$ (denoting the set of agents who have not been allocated a single good yet). 

Let $\widehat{\mu}_i$ be agent $i$'s maximin share in this reduced instance, $\widehat{\mu}_i := \max_{(B_1, \ldots, B_{|A|}) \in \Pi_{|A|} (G) }  \min_{ j  \in [|A|] } v_i( B_j) $. Next we will show that $\widehat{\mu}_i \geq \mu_i $ and $v_i(P_i) \geq \frac{1}{3}\left(1 - \frac{1}{e} \right)  \widehat{\mu}_i$. The penultimate inequality and the assumption $\tau_i \leq \mu_i$ imply that $\tau_i \leq \widehat{\mu}_i$. Therefore, establishing an approximation guarantee in terms of $\widehat{\mu}_i$ (i.e., finding a $|A|$-partition $(P_1, \ldots, P_{|A|})$ of $G$ such that $v_i(P_i) \geq \frac{1}{3}\left(1 - \frac{1}{e} \right)  \widehat{\mu}_i$ for all $i \in A$) will prove the lemma. Let $ \mathcal{M}^* = (M_1^*, \ldots, M_n^*) $ denote a partition that achieves the maximin share with respect to $v_i$ in the original instance: $ \mathcal{M}^* \in \argmax_{(B_1, \ldots, B_n) \in \Pi_n(m)} \ \min_{j \in [n]} v_i (B_j)$. Since $|[m]\setminus G| = n - |A|$, there are at most $n-|A|$ subsets $M^*_j$s which intersect with $[m]\setminus G$. In other words, there are at least $|A|$ bundles in $\mathcal{M}^*$ that are completely contained in $G$. Recall that $v_i(M_j^*) \geq \mu_i$ for all $j \in [n]$ and $v_i$ is monotone. Hence, we can partition $G$ into $|A|$ sets, say $(M_1, \ldots, M_{|A|})$, such that $v_i(M_j) \geq \mu_i$ for all $j \in [|A|]$. This proves that $\widehat{\mu}_i \geq \mu_i$. Therefore, in the rest of the proof we focus only on the reduced instance and, with slight abuse of notation, set $n = |A|$ and $m = |G|$. 

We will now analyze the second while-loop of $\ALGS$, which allocates the goods in $G$ to agents in $A$. A complete execution of the inner for-loop (Steps 7-10) will be called a \emph{round}. So, in every round, each agent $i$ is allocated exactly one good.\footnote{Without loss of generality, we can assume that the number of goods in $G$ is a multiple of $|A|$. We can do this by adding dummy goods of value $0$.} 
	
In the remainder of the proof, $f$ will be used to denote $v_i$ and $F$ to denote the multilinear extension of $f$. Write $\mu$ to denote $\widehat{\mu}_i$ and $u \in [0,1]^m$ to denote the uniform vector $\left( \frac{1}{n}, \frac{1}{n}, \ldots, \frac{1}{n} \right)$. For a particular round $r$, we will use $P^{r}$ to denote the set of goods allocated to the agent $i$ by the end of round $r$ (i.e., before round $r+1$), $G^{r}$ to denote the set of goods that remain unallocated after round $r$ (equivalently before round $r+1$), $g^r$ to denote the good allocated to agent $i$ in the round $r$, and $L^r$ to denote the set of goods allocated during $r$.  

Similarly, $G^{r+1}$ is the set of unallocated goods after round $r+1$ and $P^{r+1}$ is the bundle assigned to the agent after round $r+1$. Note that $G^{r+1} = G^r \setminus L^{r+1}$, $P^{r+1} = P^r \cup \{ g^{r+1} \}$; where $L^{r+1}$ is the set of all the goods assigned during round $r+1$ and $g^{r+1} \in L^{r+1}$ is the good assigned to agent $i$ in that round. We will follow the convention of denoting the set of goods that are assigned in the second while-loop of the algorithm by $G^0 \coloneqq G$. 

To account for the gain in valuation of the agent across rounds, we will consider the (expected) marginal value of a uniformly random subset of $G^{r}$ (i.e., the expected marginal value of a random subset of the goods that remain unallocated after round $r$); specifically, write 
\begin{align*}
E^{r} := F_{P^{r}}\left( u({G^{r}}) \right) = \E_{S \sim u({G^{r}})} \left[ f_{P^{r}} (S) \right].
\end{align*} 

Recall that $u= \left( \frac{1}{n}, \frac{1}{n}, \ldots, \frac{1}{n} \right)$ and the $g$th component of the projected vector $u(G^r)$ is equal to $1/n$ if $g \in G^r$ and this component is zero otherwise. 

Let $E^0 := F(u({G^0}))$ and note that Lemma~\ref{lem:prop} gives us $E^0 \geq \left( 1- \frac{1}{e} \right)\mu $. {Initially $P^0 = \emptyset$ and, hence, $E^0 = F_{P^0} \left( u({G^0}) \right)$.} 

We will now provide a bound for $E^r$ in terms of $E^{r+1}$, which will lead us to a telescoping sum. 

\begin{align*}
E^{r} & = \E_{S \sim u({G^{r}})} \left[ f_{P^{r}} (S) \right] \\
& \leq \E_{S \sim u({G^{r}})} \left[ f_{P^r} ( S \cap G^{r+1} ) + f_{P^r} ( S \cap L^{r+1}) \right]  \tag{since $G^r = G^{r+1} \cup L^{r+1}$ and $f_{P^r}$ is submodular} \\
& = \E_{S \sim u({G^{r+1}})} \left[ f_{P^r} ( S) \right] + \E_{S \sim u({L^{r+1}})} \left[ f_{P^r} ( S) \right]  \tag{via the definition of $u({G^{r+1}})$ and $u({L^{r+1}})$} \\
& = F_{P^r} \left( u({G^{r+1}}) \right) +  F_{P^r} \left( u(L^{r+1}) \right)  \tag{using the definition of multilinear extensions} 
\end{align*}

Applying Propostion~\ref{claim:main-ineq2} (with $P = P^r$, $g = g^{r+1}$, and $x =u(G^{r+1})$) we have $ F_{P^r} \left( u(G^{r+1}) \right) \leq F_{P^{r+1}} \left( u(G^{r+1}) \right) + f_{P^r} {(g^{r+1})}$. Note that $E^{r+1} = F_{P^{r+1}} \left( u(G^{r+1}) \right)$. Therefore, using the abovementioned bound for $E^r$, we get 
\begin{align}
E^r \leq E^{r+1} +  f_{P^r} {(g^{r+1})} + F_{P^r} \left( u(L^{r+1}) \right) \label{ineq:intermediate}
\end{align}

The trailing term $F_{P^r} \left( u(L^{r+1}) \right)$ can be upper bounded as follows 
\begin{align}
F_{P^r} \left( u(L^{r+1}) \right) & = \E_{S \sim u(L^{r+1})} \left[ f_{P^r} (S) \right] \notag \\
& \leq \E_{S \sim u(L^{r+1})} \left[ \sum_{g \in L^{r+1}} f_{P^r} ( S \cap \{g \})  \right]  \tag{since $f_{P^r}$ is submodular} \\
& = \sum_{g \in L^{r+1}} f_{P^r}(g) \ \frac{1}{n}  \label{ineq:interim2}
\end{align}
The last inequality follows from the fact that, the probability of drawing any $g \in L^{r+1}$ under $u(L^{r+1})$ (i.e., the $g$th component of $u(L^{r+1})$) is equal to $1/n$.   

To bound $E^r - E^{r+1}$ for all rounds $r \geq 1$, note that for all goods $g$, the submodularity of $f$ implies $f_{P^r}(g) \leq f_{P^{r-1}} (g)$. Furthermore, using the fact that agent $i$ is assigned a good with the maximum possible marginal value in each round $r \geq 1$, we get $f_{P^{r-1}} (g) \leq  f_{P^{r-1}}(g^r)$ for each good $g \in L^{r+1}$.\footnote{The goods in the set $L^{r+1}$ were unallocated when agent $i$ was assigned $g^r$ in round $r$.}  Therefore, equations (\ref{ineq:interim2}) and $|L^{r+1}|=n$ imply
\begin{align} 
F_{P^r} \left( u(L^{r+1}) \right) \leq f_{P^{r-1}}(g^r) \label{ineq:interim3}
\end{align}

Overall, inequalities (\ref{ineq:intermediate}) and (\ref{ineq:interim3}) establish the following for all rounds $r \geq 1$ 
\begin{align}
E^r - E^{r+1} & \leq f_{P^r} {(g^{r+1})} + f_{P^{r-1}}(g^r)  \notag \\
& = \left(f(P^{r+1}) - f(P^r) \right) + \left(  f(P^r)  - f(P^{r-1}) \right) \notag \\
&= f(P^{r+1}) - f(P^{r-1}) \label{ineq:telescope}
\end{align}

Next we address the $r=0$ case. The preprocessing performed in the first while-loop of $\ALGS$ ensures that $f(g) \leq \frac{1}{3} \left( 1  - \frac{1}{e}\right) \tau_i \leq \frac{1}{3} \left( 1  - \frac{1}{e}\right) \mu $ for all $g \in G^0$. Using this bound, along with inequalities (\ref{ineq:intermediate}) and  (\ref{ineq:interim2}), for round $r=0$ we get\footnote{Recall that $P^0 = \emptyset$. In addition, $|L^{r+1}| = n$; as mentioned previously, we can assume, without loss of generality, that the number of goods $m$ is an integral multiple of $n$ and, hence, each agent gets a good in every round.}

\begin{align}
E^0 - E^1 & \leq f(g^1) + \frac{1}{3} \left( 1  - \frac{1}{e}\right) \mu \notag \\
& = f(P^1) + \frac{1}{3} \left( 1  - \frac{1}{e}\right)\mu \label{ineq:tele1}
\end{align}

Let $T$ denote the total number of rounds executed by $\ALGS$. By the end of the $T$th round all the goods are allocated, hence $G^{T} = \emptyset$ and $E^T = 0$. Therefore, considering a telescoping sum we get
\begin{align}
E^0 &= \sum_{r=0}^{T-1} E^r - E^{r+1} \notag \\
& = \frac{1}{3} \left( 1  - \frac{1}{e}\right)\mu + f(P^1) + \sum_{r=1}^{T-1} \left( f(P^{r+1}) - f(P^{r-1}) \right) \tag{via inequalities (\ref{ineq:tele1}) and (\ref{ineq:telescope})} \\
& = \frac{1}{3} \left( 1  - \frac{1}{e}\right)\mu + f(P^T) + f(P^{T-1}) \notag \\
& \leq 2 f(P^T) + \frac{1}{3} \left( 1  - \frac{1}{e}\right)\mu \label{ineq:end}
\end{align}

Write $\widehat{P}$ to denote the bundle assigned to agent $i$ by $\ALGS$. To recap, we have proved that when agent $i$ is assigned a good in the first while-loop of $\ALGS$, then the stated bound holds, i.e., $f(\widehat{P})\geq  \frac{1}{3}\left( 1 - \frac{1}{e}\right) \tau_i$. 

On the other hand, when agent $i$ is not assigned a good in the first while-loop of $\ALGS$, then $\ALGS$ terminates with $\widehat{P} = P^T$,  which satisfies $2 f(P^T) \geq E^0 - \frac{1}{3} \left( 1  - \frac{1}{e}\right)\mu$; see inequality (\ref{ineq:end}). Lemma~\ref{lem:prop} implies that $E^0 \geq \left( 1 - \frac{1}{e} \right) \mu$. Therefore, even in this case we have $f(\widehat{P}) \geq \frac{1}{3} \left( 1  - \frac{1}{e}\right)\mu \geq \frac{1}{3} \left( 1  - \frac{1}{e}\right)\tau_i$ (recall the assumption $\tau_i \leq \mu_i$ in the lemma statement). This completes the proof.



\subsection{Proof of Theorem~\ref{thm:mms-const}}
\label{sect:proof-submod}
In \ALGF \ the initial value of $\tau_i \ (=v_i([m])) $ is guaranteed to be at least $\mu_i$. Furthermore, Lemma~\ref{lemma:ind-guarantee} ensures that, for each agent $i$, \ALGF \ never decrements $\tau_i$ below $\frac{1}{1 + \delta} \mu_i$. Therefore, when \ALGF \ terminates, for every agent $i$, the final threshold value $\tau_i$ satisfies $\tau_i \geq \frac{1}{1 + \delta} \mu_i$ and $v_i(P_i) \geq \frac{1}{3} \left( 1 - \frac{1}{e} \right) \tau_i > 0.21 \tau_i$ (since the violated set $V = \emptyset$ at termination). In other words, when executed with a small enough constant $\delta \in (0,1)$, \ALGF \ returns a partition $(P_1, P_2, \ldots, P_{n})$ which satisfies $v_i (P_i) \geq 0.21 \mu_i$ for all $ i \in [n]$, i.e., satisfies the desired approximate fairness guarantee. 


Finally, we can bound the running time of \ALGF \ by observing that initially $\tau_i = v_i([m])$ and, hence, the maximum number of times that agent $i$ is contained in set $V$ is $\log_{(1+ \delta)} \left( \frac{v_i([m])}{\mu_i} \right)$.\footnote{As stated in the algorithm, agents whose maximin share is zero can be removed from consideration.} Overall, this bound ensures that the algorithm runs in polynomial time. \\


\section{Conclusions}

The algorithms developed in this paper find allocations that are not only (approximately) fair but also \emph{sequenceable}. Sequenceable allocations are allocations that can be obtained by ordering the agents (in a sequence) and then letting them select their most valued unallocated good one after the other. Bouveret and Lema\^{i}tre~\cite{bouveret2016efficiency} have studied this notion as an efficiency measure. Recall that our algorithm  for additive valuations (both the goods and the chores case) uses the reduction of Bouveret and Lema\^{i}tre~\cite{bouveret-basic}. This reduction explicitly provides a sequence (over the agents), following which we assign goods (with maximum possible marginal values) to the agents. Hence, the allocations computed in the additive setting are sequenceable. The algorithm developed for submodular valuations also induces a sequence over the agents: agents that participate in the preprocessing step appear exactly once in the sequence and can select their most valued unallocated good. The remaining agents are repeated in the sequence---one after the other---multiple times, since the remainder of the algorithm follows a round-robin procedure. It is shown in~\cite{bouveret2016efficiency} that  every \emph{Pareto-optimal} allocation is also sequenceable. These observations lead to the open, interesting question of whether we can efficiently compute allocations which are both (approximately) fair and Pareto optimal. Note that the existence of such allocations follows from the fact that any allocation that Pareto dominates an $\alpha$-approximate maximin fair allocation is also $\alpha$-approximate maximin fair. 

Focusing on additive valuations, the work of Amanatidis et al.~\cite{amanatidis2016truthful} studies maximin fair division in a strategic setting.  Understanding if the ideas developed in this paper---along with~\cite{ghodsi2017fair}---can be used to address maximin fair division among strategic agents with, say, submodular valuations remains an interesting direction for future work.


\section*{Acknowledgements}
Siddharth Barman gratefully acknowledges the support of a Ramanujan Fellowship (SERB - {SB/S2/RJN-128/2015}) and a Pratiksha Trust Young Investigator Award. Sanath Kumar Krishnamurthy gratefully acknowledges the support of the Akiko Yamazaki and Jerry Yang Engineering Fellowship.

\bibliographystyle{alpha}
\bibliography{mainbib}

\appendix


\section{Fair Division of Chores} \label{proof:chores}

In this section we study the fair division of chores (negatively valued goods) under additive valuations.

\subsection{Notation and Preliminaries}
Notations here are similar to the ones used in Section~\ref{sect:additive}. Write $[n]=\{1, 2, \ldots, n\}$ to denote the set of agents and $[m]=\{1,2, \ldots, m\}$ to denote the set of indivisible chores. The valuation function of an agent $i$ for a subset of chores $S \subseteq [m]$ is denoted by $v_i(S)$. We assume that agents have additive valuations. That is, for each agent $i \in [n]$ 
\begin{align*}
v_i(S) & \coloneqq \sum_{g \in S} v_i(g).
\end{align*}  
Here, $v_i(g)$ is the value of chore $g \in [m]$ for agent $i \in [n]$. Since we are allocating chores among the agents, we will assume that the valuation of any chore is nonpositive, $v_i(g)\leq 0$ for all $i\in [n]$ and $g\in [m]$.

Our fairness guarantee is in terms of maximin shares. Formally, 

\begin{definition}[Maximin Share]
	\label{defn:maxmin:chores}
	For an agent $i \in [n]$ and a subset of chores $S \subseteq [m]$, the $n$-maximin share is defined to be:
	\begin{align}
	\mu_i^n(S) \coloneqq \max_{ (M_1, M_2, \ldots, M_n) \in \Pi_n(S)} \  \min_{k \in [n]} v_i(M_k).
	\end{align}
\end{definition}

Ideally, we would like to ensure fairness by partitioning the chores such that each agent gets her maximin share, i.e., partition the chores into subsets  $(A_1, A_2, \ldots, A_n) \in \Pi_n(m)$ such that  $v_i(A_i) \geq \mu_i$ for all $ i \in [n]$. Since such partitions do not always exist \cite{aziz-chores}, a natural goal is to study approximation guarantees. In particular, our objective is to develop efficient algorithms that determine a partition $(A_1, \ldots, A_n) \in \Pi_n(m)$ wherein each agent, $i$, gets a bundle, $A_i$, of value (under $v_i$) at least $\alpha$ times her maximin share, with $\alpha \in [1,\infty)$ being as small as possible. We call such partitions as $\alpha$-approximate maximin fair allocations. When $\alpha=1$, we say that the allocation is maximin fair. Observe that in the context of chores, our approximation factors are greater than 1. This is because, in this case, the maximin shares are nonpositive.

We present an efficient algorithm that finds $4/3$-approximate maximin fair allocations of chores. This result improves upon the $2$-approximation bound of \cite{aziz-chores}.

\subsection{Approximation Guarantee for Chores}
Using techniques similar to the case of goods, we prove an approximate maximin shares guarantee for fair division of chores with additive valuations.

\begin{theorem}[Main Result for Chores]
\label{thm:main-chores}
Given a set of $n$ agents that have additive valuations, $\{v_i\}_{i \in [n]}$, for a set of $m$ indivisible chores (i.e., $v_{i}(g) \leq 0$ for all $i \in [n]$ and $g \in [m]$), we can find, in polynomial time, a partition $(A_1, \ldots, A_n) \in \Pi_n(m)$ that satisfies 
\begin{align}
v_i(A_i) \geq \frac{4n-1}{3n} \ \mu_i \qquad \text{ for all } i \in [n].
\end{align}
Here, $\mu_i$ is the maximin share of agent $i$. 
\end{theorem}

The proof of Theorem~\ref{thm:main-chores} proceeds in two parts:
\begin{enumerate}
\item[(i)] First, we observe that the problem of finding a partition that provides approximate maximin fair allocations can be reduced to a restricted setting wherein the agents value the chores in the same order. This follows from the observation that the arguments used in Section~\ref{section:bouveret} hold even in the case of chores--the arguments in Section~\ref{section:bouveret} rely on the additivity of the valuations and do not require the goods to be positively valued. Hence, we can solely focus on the case in which the $m$ chores can be ordered (indexed), say $g_1, g_2, \ldots, g_m$, such that the valuation of every agent $i$ respects this ordering: for each  $a < b$ we have $ v_{i} (g_a) \leq v_{i} (g_b)$.
\item[(ii)] Then, we develop an efficient algorithm that achieves a $4/3$-approximation guarantee for this restricted/ordered setting. 
\end{enumerate}

\subsection{Envy Graph Algorithm for Chores}
Since we only need to address the setting in which the agents value the chores in the same order, throughout this section we will assume that the chores are indexed, say $g_1, g_2, \ldots, g_m$, such that for every agent $i$ we have $ v_{i} (g_a) \leq v_{i} (g_b)$ (i.e., $ |v_{i} (g_a)| \geq |v_{i} (g_b)|$), for all indices $ a<b$. 

Our approximation algorithm iteratively allocates the chores in increasing order of their indices and maintains a partial allocation, $\A=(A_1, \ldots, A_n)$, of the chores assigned so far. In order to assign a chore the algorithm selects a bundle, $A_i$, by considering a directed graph, $G(\A)$, that captures the envy between agents. The nodes in this envy graph represent the agents and it contains a directed edge from $i$ to $j$ iff $i$ envies $j$, i.e., $v_i(A_i) < v_i(A_j)$. 

We restate Lemma~\ref{lemma:envy} for the case of chores; the proof of the following lemma is identical to that of Lemma~\ref{lemma:envy}. 

\begin{lemma}[\cite{lipton-envy-graph}]
	\label{lemma:envy:chores}
	Given a partial allocation $\A=(A_1, \ldots, A_n)$ of a subset of chores $S \subseteq [m]$, we can find another partial allocation $\B=(B_1, \ldots, B_n)$ of $S$ in polynomial time such that 
	\begin{enumerate}
		\item[(i)] The valuations of agents for their bundles do not decrease: $v_i(B_i) \geq v_i(A_i)$ for all $i \in [n]$.
		\item[(ii)] The envy graph $G(\B)$ is acyclic.
	\end{enumerate}
\end{lemma}

{
\begin{algorithm}
{
	{\bf Input :} $n$ agents, $m$ indivisible chores, and valuations $v_{i,g}$ for each $i \in [n]$ and $g \in [m]$. In the given instance, the agents value the chores in the same order. \\ {\bf Output:} An approximate maximin fair allocation.
\caption{Envy Graph Algorithm $\ALGC$}
\label{alg:envygraph-chores}
\begin{algorithmic}[1]
  
\STATE Order the chores, $g_1, g_2, \ldots, g_m$, such that for every agent $i \in [n]$ and $a<b$ we have $|v_{i g_a}| \geq |v_{i g_b}|$.
\STATE Initialize allocation $\A=(A_1, A_2, \ldots, A_n)$ with $A_i = \emptyset$ for all $ i \in [n]$.
\FOR{$j =1$ to $m$}
\STATE Pick a vertex $i$ that has no outgoing edge in the envy graph $G(\A)$, i.e., $i$ is a \emph{sink vertex} in $G(\A)$. \\
\COMMENT{The algorithm maintains the invariant that $G(\A)$ is acyclic. Hence, such a vertex is guaranteed to exist.}
\STATE Update $A_i \leftarrow A_i \cup \{ g_j \}$.
\IF{the current envy graph $G(\A)$ contains a cycle} 
\STATE Use Lemma~\ref{lemma:envy:chores} to update $\A$ and, hence, obtain an acyclic envy graph. 
\ENDIF
\ENDFOR 
\STATE Return $\A$.
\end{algorithmic}
}
\end{algorithm}
}

\subsection{Proof of Theorem~\ref{thm:main-chores}}
It is clear that the $\ALGC$ runs in polynomial time. To establish the stated claim, in the remainder of this proof we will show that $v_1(A_1) \geq \frac{4n-1}{3n} \mu_1$; an analogous argument establishes the desired bound, $v_i(A_i) \geq \frac{4n-1}{3n}\mu_i$, for all other agents.

Write $\A=(A_1, \ldots, A_n)$ to denote the allocation returned by the algorithm $\ALGC$. Let $g_r$ be the last chore assigned to agent $1$ in the for-loop of the algorithm, i.e., agent $1$ was selected as a sink node in the $r$th iteration. Also, note that $g_r$ may not lie in the final bundle $A_1$---since the bundle of agent $1$ can get swapped while applying Lemma~\ref{lemma:envy:chores}. Let $\A'=(A'_1, \ldots, A'_n)$ be the allocation just after the chore $g_r$ is allocated to agent $1$; we have $g_r \in A'_1$. 

Since $g_r$ is the last chore allocated to agent $1$, Lemma~\ref{lemma:envy:chores} gives us $v_1(A_1) \geq v_1(A'_1)$ (i.e., $|v_1(A_1)| \leq |v_1(A'_1)|$). In addition, the fact that the chores are only allocated to sinks implies that when $g_r$ was allocated to agent $1$ it must have been the case that she did not envy any other agent. That is,  
$|v_1(A_1)| \leq |v_1(A'_1)| \leq |v_1(A_i)| + |v_1(g_r)|$ for all $i \neq 1$. Summing this inequality over $i$, we get:

\begin{align*}
	|v_1(A_1)| &\leq \frac{1}{n}\left[|v_1(A_1)| + \sum_{j \neq 1} (|v_1(A_j)| + |v_1(g_r)|)\right]\\
	&\leq \frac{1}{n}\left(n|\mu_1| + (n-1)|v_1(g_r)|\right)\\
	&\leq |\mu_1| + (1-\frac{1}{n})|v_1(g_r)|
\end{align*}

Here, the second inequality follows from the fact that $\sum_{i=1}^n v_1(A_i) = v_1([m]) \geq n\mu_1$ and all the involved quantities ($v_i(A_i)$, $v_1([m])$, and $\mu_1$) are nonpositive.

Suppose $|v_1(g_r)|\leq \frac{1}{3}|\mu_1|$. Then,
\begin{align*}
|v_1(A_1)| &\leq |\mu_1| + (1-\frac{1}{n})|v_1(g_r)|\\
&\leq |\mu_1| + \frac{1}{3}(1-\frac{1}{n})|\mu_1|\\
& = \frac{4n-1}{3}|\mu_1|
\end{align*}

Therefore, to prove Theorem~\ref{thm:main-chores}, we only need to consider the case when the last chore assigned to agent $1$ has value less than $\frac{1}{3}\mu_1$, i.e., $|v_1(g_r)|> \frac{1}{3}|\mu_1|$; the rest of the proof addresses this case. 

In particular, we will use the inequality $|v_1(g_r)|> \frac{1}{3}|\mu_1|$ to prove that the index $r$ is at most $2n$. Suppose for contradiction $r > 2n$, then in any partition $\Pa=(P_1,P_2,\dots,P_n)$ there exists a bundle $P_k$ with 3 chores $\{g_{i_1},g_{i_2},g_{i_3}\} \subseteq P_k$ such that $i_1,i_2,i_3 \leq r$ (by the pigeonhole principle). Therefore, $v_1(P_k) \leq v_1(\{g_{i_1},g_{i_2},g_{i_3}\}) \leq 3v_1(g_r) < \mu_1$. In particular, this is true for partitions that achieve maximin share under the valuation $v_1$, which contradicts the definition of $\mu_1$.

The proof of Lemma~\ref{lemma-LPT-chores} follows from the arguments in~\cite{graham-makespan} and we provide it here for completeness. 

\begin{lemma} \label{lemma-LPT-chores}
	Consider an additive valuation function $v$. Let $\{ p_1,p_2,\ldots, p_d \}$ be a set of $d$ chores to be partitioned among $n$ agents, with $ d \leq 2n$ and $v(p_1) \leq v(p_2) \leq \ldots v(p_d) \leq 0$. If $|v(p_d)| > |\frac{\mu}{3}|$---where $\mu$ is the $n$-maximin share under the valuation function $v$---then, the partition \begin{align*} \Pa & = \left( \{p_1\},\{p_2\},\dots,\{p_{2n-d}\},\{p_{2n-d+1},p_d\},\{p_{2n-d+2},p_{d-1}\},\dots,\{p_n,p_{n+1}\} \right)\end{align*} achieves the maximin share under the valuation $v$.
\end{lemma}
\begin{proof}
Suppose, for contradiction, that the partition $\Pa$ does not achieve the maximin share. Hence, there exists is a bundle in $\Pa$ that has value less than $\mu$. Note that every chore must be assigned to some partition, therefore $v(p_1) \geq \mu$. Hence, the bundle in $\Pa$, that has value less than $\mu$ must contain at least two chores. With this observation in hand, define index $\ell \coloneqq \min\{h \mid h \geq n \text{ and } |v(\{p_{2n-h+1},p_h\}) | \geq |\mu| \}$. 

Note that, for all $i \leq (2n- \ell +1)$ and $j\leq \ell$ such that $i\neq j$, we have $|v(\{p_i,p_j\})|=|v(p_i)|+|v(p_j)|\geq |v(p_{2n-\ell+1})|+|v(p_\ell)|=|v(\{p_{2n-\ell+1},p_\ell\})|\geq|\mu|$. 



Therefore, in any partition that achieves the maximin share, say $\mathcal{M}=(M_1, \ldots, M_n)$, the $(2n- \ell +1)$ chores $L \coloneqq \{p_1,p_2,\ldots, p_{2n-\ell+1}\}$ must not paired among themselves\footnote{That is, for all $i \in [n]$, we must have $|M_i \cap L | \leq 1$} or with any chore from the set $S \coloneqq \{p_{2n- \ell +2}, \ldots, p_\ell \}$. Hence, the $(2 \ell - 2 n +1)$ chores in the set $S$ are assigned to the remaining $(\ell-n-1)$ bundles (which do not contain a chore from $L$). This implies, by the  pigeonhole principle, that some bundle in $\mathcal{M}$ contains three chores. Such a bundle would be of value less than $3v(p_d)<\mu$. This, however, contradicts the fact that $\mathcal{M}$ achieves the maximin share $\mu$.

\end{proof}


Recall that $A'_1$ is the bundle assigned to agent $1$ right after the $r$th iteration (i.e., right after agent $1$ receives the chore $g_r$). As mentioned previously, $|v_1(A_1)| \leq |v_1(A'_1)|$. Therefore, it is enough to prove that, $|v_1(A'_1)| \leq |\mu_1|$.

We will now prove that $A'_1$ contains at most 2 chores. Suppose for contradiction $|A'_1| \geq 3$. Therefore, $|v_1(A'_1 \setminus \{g_r\})| \geq 2|v_1(g_r)|$. Recall, agent~1 was assigned the chore $g_r$ only when she did not envy any other agent. Therefore, $|v_1(A'_j)| \geq  2|v_1(g_r)|$ for all $j\in[n]$. Also, recall that the partial allocation $\A'$ is a partition of the chores $\{g_1,g_2,\dots,g_r\}$. Note that by including more chores to be this set, we can only decrease the $n$-maximin share. 

Consider an additive, set function $f$ defined over subsets of $\{g_1,g_2,\dots,g_r\}$; in particular, $f(g_k) \coloneqq \myfloor{\frac{|v_1(g_k)|}{|v_1(g_r)|}}$
and for $S \subseteq \{g_1,g_2,\dots,g_r\}$, we have $f(S) \coloneqq  \sum_{g \in S} f(g)$. By definition, for all subsets $S$, we have $|v_1(S)|\geq f(S) \ |v_1(g_r)|$.  Note that $g_r$ is the highest valued chore among $\{g_1,g_2,\dots,g_r\}$, i.e., $|v_1(g_k)| \geq |v_1(g_r)|$ for all $k\in[r]$. Therefore, we have $f(S) \geq |S|$, for every subset $S$. 

These observations imply that $f(A'_j) \geq 2 $ for all $j \in [n]$: 
\begin{itemize}
	\item If, for $j \in [n]$, we have $|A'_j| \geq 2$, then $f(A'_j) \geq |A'_j| \geq 2$.
	\item Otherwise, if for $j \in [n]$, we have $|A'_j| = 1$, then the fact that $|v_1(A'_j)| \geq 2|v_1(g_r)|$ (i.e., the single chore in $A'_j$ has absolute value at least twice that of $g_r$) gives us $f(A'_j) \geq 2$. 
\end{itemize} 

We have assumed (towards a contradiction) that $|A'_1| \geq 3$, hence $f(A'_1) \geq |A'_1| \geq 3$. Therefore, $f\left( \{g_i\}_{i=1}^r  \right) = \sum_{i=1}^n f(A'_i) \geq 2n+1$. This bound implies that in any partition $\Pa=\{P_1, \ldots, P_n\}$ of $\{g_i\}_{i=1}^r$, there exists a bundle $P_k$ with the property that $f(P_k) \geq 3$. Since $|v_1(g_r)|>\frac{1}{3}|\mu_1|$ and $|v_1(P_k)|\geq f(P_k)\cdot |v_1(g_r)|$, this implies that $v_1(P_k)$ is less than $\mu_1$. In particular, this is true for partitions that achieve maximin share under the valuation $v_1$, which contradicts the definition of $\mu_1$. Therefore, $|A'_1| \leq 2$.

Since $A'_1$ contains at most two chores, we have the following two cases:
\begin{itemize}
	\item Case {\rm I}: $|A'_1| = 1$.  That is, $A'_1 =\{ g_r \}$. Let $\mathcal{M}=\{M_1, \ldots, M_n \}$ be an allocation that achieves the maximin share for agent $1$, i.e., $\mu_1 = \min_j v_1(M_j)$. The chore $g_r$ must be contained in one of the bundles in $\mathcal{M}$, say $g_r \in M_j$. Then, $|v_1(A'_1)|
= |v_1(g_r)| \leq |v_1(M_j)| \leq |\mu_1| $. That is, agent $1$'s value for $A'_1$ is at least its maximin share. 	
	\item Case {\rm II}: $|A'_1| = 2$. Since $\ALGC$ allocates the first $n$ chores to $n$ different bundles, we have $A'_1=\{ g_\ell, g_r\}$, with $\ell \leq n$. Furthermore, we note that index $\ell \geq 2n-r+1$: when agent $1$ was assigned chore $g_r$ she did not envy any other agent. Also, at that point of time agent $1$ owned the singleton $\{g_\ell \}$. These observations imply that, by the time chore $g_r$ was assigned to agent $1$, the $(n-\ell)$ chores $\{g_{\ell+1}, \ldots, g_{n} \}$ (which, recall that, are initially assigned as singletons) must have been paired up with a chore from the $( r - n - 1)$-size set $\{g_{n+1}, \ldots, g_{r-1} \}$; if this is not the case then agent $1$ would envy a singleton bundle $\{g \}$ with $g \in \{g_{\ell+1}, \ldots, g_{n} \}$. This counting argument gives us $n- \ell \leq r - n - 1$, i.e., $\ell \geq 2n - r + 1$. 
	
	Lemma~\ref{lemma-LPT-chores} ensures that the partition $\Pa$ achieves the maximin share (under valuation $v_1$). Since this partition contains the bundle $P_s \coloneqq \{g_{2n-r+1},g_r\}$, we have $|v_1(P_s)| \leq |\mu_1|$. The fact that $\ell \geq 2n - r + 1$ (i.e., $v_1(g_\ell) \geq v_1(g_{2n-r+1})$) gives us the desired inequality $|v_1(A'_1)| = |v_1(\{g_\ell, g_r \})| \leq |v_1(P_s)| \leq |\mu_1|$.
\end{itemize}

Therefore if $|v_1(g_r)|>\frac{1}{3}|\mu_1|$, then the allocation $\A$ is maximin fair (according to agent $1$), i.e., $v_1(A_1) \geq |\mu_1|$. This completes the proof of Theorem~\ref{thm:main-chores}.

\section{Examples} \label{Examples}
\subsection{EF1 does not imply Maximin Fair} \label{example:ef1:mms}
The following simple example shows that an EF1 allocation may not have a constant-factor maximin guarantee. 

Consider an instance with $n$ agents, $2n-1$ goods, and identical valuations $v$: $v(j)=1$ for all $ 1 \leq j \leq n$ and $v(j)=n$ for all $n+1\leq j \leq 2n-1$. 
Note that allocation $\mathcal{A} =(A_1,A_2,\ldots, A_n)= (\{1\},\{2,n+1\},\{3,n+2\},\ldots,\{n,2n-1\})$ is EF1. 
Furthermore, considering allocation $\mathcal{B} =(B_1,B_2,\dots,B_n)= (\{1,2,\dots,n\},\{n+1\},\{n+2\},\dots,\{2n-1\})$ we get that the maximin share, $\mu$, under the valuation $v$ is at least $n$: $\min_{i \in [n] }  v(B_i) = n$.  

But, given that $v(A_1) = 1$, we get an EF1 allocation which does not have a constant factor maximin guarantee. 

\subsection{Nonexistence of Maximin Fair Allocation} \label{example:mms:existance}
The following simple example shows an instance where valuations are submodular and a $(0.75 + \epsilon)$ approximate maximin fair allocation does not exist (for any $\epsilon>0$).

Consider an instance with 2 agents, 4 goods ($\{a_1,a_2,b_1,b_2\}$). Agent~1's valuation is denoted by $v_1$ and defined by: $v_1(S)=1$ for all $|S|=1$; $v_1(\{x_1,x_2\})=2$ for all $x \in \{a,b\}$; $v_1(\{a_i,b_j\})=1.5$ for all $i,j\in\{1,2\}$; $v_1(S)=2.5$ for all $S$ such that $|S|=3$; $v_1(\{a_1,a_2,b_1,b_2\})=3$. Agent~2's valuation is denoted by $v_2$ and defined by: $v_2(S)=1$ for all $|S|=1$; $v_2(\{a_j,b_j\})=2$ for all $j \in \{1,2\}$; $v_2(\{x_1,y_2\})=1.5$ for all $x,y\in\{a,b\}$; $v_2(S)=2.5$ for all $S$ such that $|S|=3$; $v_2(\{a_1,a_2,b_1,b_2\})=3$. One can easily check the submodularity of $v_1$ and $v_2$.

According to agent~1, there is only one unique maximin fair partition $\Pa=(\{a_1,a_2\},\{b_1,b_2\})$ and agent~1's maximin share is 2. Similarly, according to agent~2, there is only one unique maximin fair partition $\Pa'=(\{a_1,b_1\},\{a_2,b_2\})$ and agent~2's maximin share is also 2. For any allocation $\A$, if agent~1's value for his own bundle is at least 2, then agent~2's value for his own bundle is atmost 1.5. And, for any $S\subseteq\{a_1,a_2,b_1,b_2\}$ and $i\in\{1,2\}$, we know that $v_i(S)\geq 2$ or $v_i(S)\leq 2$. Therefore, in this instance, it is not possible to construct a $(0.75+\epsilon)$ approximate maximin fair allocation (for any $\epsilon>0$). 

It is relevant to note that if we have two agents with additive valuations, then a maximin fair allocation always exists. This example shows that such a guarantee does not hold for submodular valuations.

\end{document}